\documentclass[12pt]{article}
\usepackage{amssymb}
\usepackage{amsmath}
\usepackage{amsfonts}
\usepackage{amsthm}
\usepackage{cite}
\usepackage{color}

\newcommand{\bea}{\begin{eqnarray}}
	\newcommand{\eea}{\end{eqnarray}}

\newcommand{\be}{\begin{equation}}
	\newcommand{\ee}{\end{equation}}

\theoremstyle{plain}
\newtheorem{theorem}{Theorem}

\newtheorem{lemma}[theorem]{Lemma}

\theoremstyle{definition}
\newtheorem{definition}[theorem]{Definition}

\newcommand{\beas}{\begin{eqnarray*}}
	\newcommand{\eeas}{\end{eqnarray*}}

\numberwithin{equation}{section}
\numberwithin{theorem}{section}

\begin{document}
	
	\centerline{\Large {\bf Stratified Lie systems: Theory and applications}}
	\vskip 0.5cm
	
	\centerline{ J.F. Cari\~nena$^\dagger$, J. de Lucas$^\ddagger$, and D. Wysocki$^\ddagger$}
	\vskip 0.5cm
	
	\centerline{$^\dagger$Departamento de  F\'{\i}sica Te\'orica, Universidad de Zaragoza,}
	\medskip
	\centerline{P. Cerbuna 12, 50.009 Zaragoza, Spain.}
	\medskip
	\centerline{$^\ddagger$Department of Mathematical Methods in Physics, University of Warsaw,}
	\medskip
	\centerline{ul. Pasteura 5, 02-093 Warszawa, Poland}
	\medskip
	
	\vskip 1cm
	
	\begin{abstract}{
			A {\it stratified Lie system} is a nonautonomous system of first-order ordinary differential
			equations on a manifold $M$ described by a $t$-dependent vector field  $X=\sum_{\alpha=1}^rg_\alpha X_\alpha$, where $X_1,\ldots,X_r$ are vector fields on $M$ spanning an $r$-dimensional Lie algebra that are tangent to the strata of a stratification $\mathcal{F}$ of $M$ while $g_1,\ldots,g_r:\mathbb{R}\times M\rightarrow \mathbb{R}$ are functions depending on $t$ that are constant along integral curves of $X_1,\ldots,X_r$ for each fixed $t$. 
			We analyse the particular solutions 
			of stratified Lie systems and how their properties can be obtained as generalisations of those of Lie systems. We illustrate our results by studying Lax
			pairs and a class of $t$-dependent Hamiltonian systems. We study stratified Lie systems with compatible geometric structures. In particular, a class of stratified Lie systems on Lie algebras are studied via Poisson structures induced by $r$-matrices.
		}
		
		
		
		
	\end{abstract}

	\section{Introduction}
	A {\it Lie system} is a nonautonomous  system of first-order ordinary differential equations (ODEs) in normal form
	whose general solution can be written as a function, a so-called
	{\it superposition rule}, of a family of particular solutions and
	some constants related to initial conditions \cite{CGM00,CGM07,Dissertationes,LS,PW}. The Lie--Scheffers theorem \cite{CGM07,LS,PW,LS20} states that a Lie system
	amounts to a $t$-dependent vector field taking values in a finite-dimensional Lie algebra
	of vector fields, called the {\it Vessiot--Guldberg Lie algebra}  of the Lie system.
	
	By applying the Lie--Scheffers theorem, Lie proved that
	every Lie system on the real line is locally diffeomorphic, around a generic point, to
	a  Riccati differential equation\footnote{Some works additionally assume that $a_1(t)a_3(t)$ must not be equal to zero for every $t\in \mathbb{R}$ (see \cite{Ri24}).}
	$$
	\frac{dx}{dt}=a_1(t)+a_2(t)x+a_3(t)x^2,
	$$
	where $a_1(t),a_2(t),a_3(t)$ are arbitrary $t$-dependent functions \cite{Dissertationes,LS,LS20}.
	Although Lie also classified all finite-dimensional Lie algebras of vector fields on the plane around a generic point up to local diffeomorphisms, his results presented several unclear points (see \cite{GKO92}).
	Gonz\'alez--L\'opez, Kamran, and Olver clarified Lie's classification and, as a result, they proved that there exist 28 families of finite-dimensional Lie algebras of vector fields on $\mathbb{R}^2$, the hereafter called {\it GKO classification}. From their results and the Lie--Scheffers theorem, one can classify, up to local diffeomorphisms, Lie systems at generic points of the plane \cite{BBHLS15,GKO92}. 
	
	Previous facts illustrate that most systems of differential equations are not Lie
	systems \cite{CGL08,Dissertationes}. Notwithstanding, Lie systems have
	a plethora of geometric properties and relevant applications \cite{CGM00,Dissertationes,LS20,LV15,PW},
	e.g. matrix Riccati equations are Lie systems appearing in
	the study of B\"acklund transformations and other fields \cite{CR08,CRF01,OlmRodWin86,OlmRodWin87,PW}, which motivates their analysis.
	
	The theory of Lie systems has been extended in different manners to analyse much more general families of systems of differential equations. {\it PDE Lie systems} \cite{CGM07,OG00,Ra11} were applied to the study of conditional symmetries and B\"acklund transformations  \cite{CGL19,LG18}. {\it Quasi-Lie schemes} and {\it quasi-Lie systems}
	were developed to investigate integrability conditions for systems of ODES and partial differential equations (PDEs), e.g. dissipative Milne-Piney equations and nonlinear oscillators \cite{CGL08,CGL19,CLL09}. Superposition rules for discrete differential equations were considered by Winternitz and his collaborators \cite{PW04,RW84,RSW97}. Super-superposition rules, which are aimed at the analysis of general solutions to superdifferential equations, were  analysed in \cite{BGHW87,BGHW90}. A detailed survey on the previous and other generalisations of the theory of Lie systems can be found in \cite{Dissertationes,LS20}.

	This work focuses on another generalisation of Lie systems
	that has been scarcely analysed so far: the  {\it foliated Lie systems} \cite{CGM00}, which are here more properly called {\it stratified Lie systems}. 	Recall that a {\it stratification} of a manifold $N$ is a partition of $N$ into connected disjoint immersed submanifolds, the so-called {\it strata}, of not necessarily the same dimension \cite{OR04,St80,Su73,Va94}. If all the strata have the same dimension, the stratification is said to be {\it regular}. In this case, it is said that the stratification is a {\it foliation} and its strata are called {\it leaves}. 
	
		A {\it stratified Lie system}  
	is a nonautonomous  system
	of first-order ODEs in normal form describing the integral curves of a $t$-dependent vector field on a manifold $N$ of the form
	\begin{equation}\label{De}
		X(t,x) = \sum_{\alpha=1}^rg_\alpha(t,x)X_\alpha(x),\qquad \forall t\in \mathbb{R},\quad\forall x\in N,
	\end{equation}
	where  $X_1,\ldots,X_r$ are vector fields on $N$ that span an $r$-dimensional Lie algebra $V$, i.e. $[X_\alpha,X_\beta]=\sum_{\gamma=1}^rc_{\alpha\beta\gamma}X_\gamma$ for certain real constants $c_{\alpha\beta\gamma}$ with $\alpha,\beta,\gamma=1,\ldots,r$, and they therefore span an integrable Stefan--Sussmann distribution $\mathcal{D}^V=\{Y(x):Y\in V,x\in N\}$ on $N$ (see \cite{La18,St80,Su73,Va94} for details), while
	$g_1,\ldots,g_r\in C^\infty(\mathbb{R}\times N)$ are common $t$-dependent constants of motion of  $X_1,\ldots,X_r$, namely if we consider $X_1,\ldots,X_r$  as vector fields on $\mathbb{R}\times N$ in the natural way  \cite{Dissertationes}, then $X_\alpha g_\beta=0$ for $\alpha,\beta=1,\ldots,r$.  Finally, if $g_1,\ldots,g_r$ depend only on time, (\ref{De}) is called a {\it Lie system}.

An integrable Stefan-Sussmann distribution on $N$, like $\mathcal{D}^V$ (see \cite{La18,St80,Su73,Va94}), gives rise to a stratification so that the tangent bundles to their strata are determined by the Stefan-Sussman distribution. If $\mathcal{F}$ is the stratification by integral submanifolds induced by $\mathcal{D}^V$, we call $X$ and $V$ an {\it $\mathcal{F}$-stratified Lie system} and a {\it Vessiot--Guldberg Lie algebra} of $X$, respectively. The elements of $V$ are tangent to the leaves of $\mathcal{F}$. It follows from (\ref{De}) that each vector field $X_t:x\in N\mapsto X(t,x)\in TN$, for every fixed $t\in \mathbb{R},$ is also tangent to the strata of $\mathcal{F}$. Hence, the restriction of $X$ to any stratum of $\mathcal{F}$ is a Lie system admitting as a Vessiot--Guldberg Lie algebra the restriction of $V$ to the stratum.
	
Now it is clear that although our definition of stratified Lie system matches what is called a {\it foliated Lie system} in \cite{CGM00}, the change of terminology is due to the fact that (\ref{De}) is, more precisely, associated with a stratification than with a foliation. In fact, the latter only occurs when the strata of the stratification of $\mathcal{D}^V$ have all the same dimension. Anyhow, we shall show that assuming that $\mathcal{D}^V$ gives rise to a foliation is a very mild condition and allows for studying relevant problems while avoiding technical minor details. In fact, most results in this paper are proved for stratified Lie systems whose associated stratification is a foliation.
	   
	The work \cite{CGM00} provided a few applications of stratified Lie systems to the theory of integrable systems and many other theoretical examples. Apart from \cite{LZ21}, where stratified Lie systems are applied to describe relative equilibrium points of a $t$-dependent energy-momentum method, no new result on stratified Lie systems seems to have been analysed in the literature. Meanwhile, our work provides new theoretical results and physical applications of stratified Lie systems. 
	

	We prove that foliated Lie systems appear naturally while transforming a $t$-dependent Hamiltonian system onto a new simpler one through a $t$-dependent canonical transformation \cite{Va15}. It is also shown that such foliated Lie systems admit a Lax pair formulation, providing a nonautonomous generalisation of results given by Babelon and Viallet in \cite{BV90}. 
	
	Then, we define an $\mathcal{F}$-{\it foliated superposition rule} notion for a nonautonomous system of first-order ODEs in normal form on $N$ as a function $\Psi:N^{m+1}\rightarrow N$ such that $\Psi(\mathcal{F}_k^{m+1})\subset \mathcal{F}_k$ for every leaf $\mathcal{F}_k$ of $\mathcal{F}$, and $\Psi$ allows us to describe the particular solutions of the system passing through any leaf of $\mathcal{F}$ in terms of a generic  family of $m$ particular solutions contained in the same leaf and a parameter in $\mathcal{F}_k$ related to the initial conditions of each particular solution in $\mathcal{F}_k$. As an application, we provide foliated superposition rules for certain Lax pairs (understood as systems of ODEs in normal form in a natural way) and a class of $t$-dependent Hamiltonian systems.
	
		It is here proved that a foliated Lie system  admits a foliated superposition rule. As a byproduct, we provide an analogue of the Lie's condition for foliated Lie systems.
	We also devise a method to obtain foliated superposition rules by means of a modification of the technique 
	developed in \cite{CGM07} to derive superposition rules for Lie systems. Additionally, our work studies the properties of first-order systems of ordinary differential equations in normal form admitting a foliated superposition rule. 
	Next, we prove that solving a foliated Lie system reduces to integrating a type of foliated Lie system on a trivial principal bundle, called an  {\it automorphic foliated Lie system}.
	In turn, given an automorphic foliated Lie system on the total space $P$ of a trivial principle bundle $\pi:P=G\times N\rightarrow N$
	, a Lie group action $\varphi:G\times M\rightarrow M$ whose set of orbits, $M/G$, can be endowed with a differentiable structure diffeomorphic to  $N$, i.e. $M/G
	\simeq N$,  allows one to construct a stratified Lie system on $M$ whose general solution can be determined by  a particular solution of the automorphic foliated Lie system within each fibre 
	of the principal bundle on $P$. These results constitute the generalisation to stratified Lie systems of the relations between standard Lie systems and automorphic Lie systems appearing in  \cite{CGM00,Dissertationes,LS20,Ve93,Ve99}.

	Finally, foliated Lie systems are employed to provide a  new generalisation of Ermakov systems admitting a Lewis-Riesenfeld invariant \cite{Ra80,RR79,RR80}. We also prove that a class of Lax pairs and their associated $t$-dependent Hamiltonian systems are related to the same automorphic foliated Lie system. As a last application, it is shown how $r$-matrices and several associated Poisson brackets on Lie algebras can be applied to the study of foliated Lie systems related to Lax pairs and automorphic foliated Lie systems. Moreover, $r$-matrices are also utilised to study certain stratified Lie systems. These examples allow us to define the so-called {\it stratified Lie--Hamilton systems}, which generalise standard {\it Lie--Hamilton systems} \cite{CLS13}. The generalisation of the theory of Lie systems with compatible geometric structures (see \cite{LS20}) can be developed analogously.
	
	The structure of the paper goes as follows. In Section 2 we survey
	the theory of Lie systems. Section 3 introduces the definition of stratified and foliated
	Lie systems by providing several new examples. 
	Examples of foliated Lie systems are studied in Section 4. Section 5 introduces foliated superposition rules, while Section 6 shows that foliated Lie systems admit a foliated
	superposition rule, gives an algorithm to derive it, and it analyses the properties of a system admitting a foliated superposition rule. This can be understood as the extension to foliated Lie systems of the Lie--Scheffers theorem. We define automorphic foliated Lie systems and explain how they can be used to solve foliated Lie systems in Section \ref{Sec:AFLS}. Section 8 develops several applications of our methods. In Section 9
	we summarise our achievements and describe some further work in progress.
	
	\section{Fundamentals of Lie systems}\label{Intro}
	Let us survey the basic theory of Lie systems and related notions needed to understand the results of our paper.
	To simplify our presentation and to stress our main
	results, we assume all mathematical structures to be smooth and globally defined (see \cite{CGM00,Dissertationes,LS20} for further details). If not otherwise stated, every
 differential equation is hereafter considered nonautonomous and every manifold is connected. In what follows, $N$ is an $n$-dimensional manifold.
	
	We call {\it generalised} or {\it Stefan-Sussmann distribution} $\mathcal{D}$ on $N$ a correspondence mapping each point $x\in N$ to a subspace $\mathcal{D}_x\subset T_xN$ (see \cite{La18,OR04,St80,Su73} for details). If the dimension of $\mathcal{D}_x$ is the same at every  $x\in N$, then $\mathcal{D}$ is called {\it regular}. In some works, a regular Stefan-Sussmann distribution is simply called a {\it distribution}. If there exists a decomposition of $N$ as a sum of disjoint connected immersed submanifolds $\mathcal{F}_\lambda$ such that  $T_x\mathcal{F}_\lambda=\mathcal{D}_x$ for every $x\in \mathcal{F}_
\lambda$ an each $\mathcal{F}_\lambda$, then $\mathcal{D}$ is said to be {\it integrable}. To simplify the terminology, generalised distributions will just be hereafter called  distributions.
	
	Let us define $\pi_2:(t,x)\in \mathbb{R}\times N\mapsto x\in N$ and let  $\tau_N:TN\rightarrow N$ be the tangent bundle projection.
	A {\it $t$-dependent vector field} on $N$ is a mapping $X:(t,x)\in \mathbb{R}\times N\mapsto X(t,x)\in TN$ such that $\tau_N\circ X=\pi_2$.  An {\it integral curve} of $X$ is a particular solution $\gamma:\mathbb{R}\rightarrow N$ of
	\begin{equation}\label{Asso}
		\frac{dx}{dt}=X(t,x),\qquad \forall (t,x)\in \mathbb{R}\times N.
	\end{equation}
	Consequently, $\widetilde{\gamma}:t\in \mathbb{R}\mapsto (t,\gamma(t))\in \mathbb{R}\times N$ is an integral curve of the {\it autonomisation} (or {\it suspension}) of $X$, i.e. the vector field  $\widetilde{X}$ on $\mathbb{R}\times N$ given by $\widetilde{X}=\partial_t+X$, where we use  the natural diffeomorphism $\widetilde{\varphi}:T(\mathbb{R}\times N)\simeq T\mathbb{R}\times TN$ \cite{AM87,Dissertationes}. The other way around, if $\widetilde{\gamma}:\mathbb{R}\rightarrow \mathbb{R}\times N$ is an integral curve of $\widetilde{X}$ and a section of the bundle $\pi_1:(t,x)\in \mathbb{R}\times N\mapsto t\in \mathbb{R}$, then $\pi_2\circ \widetilde{\gamma}$ is a solution to (\ref{Asso}). This one-to-one correspondence permits us to identify system  (\ref{Asso}) and its associated $t$-dependent vector field $X$. In turn, this allows us to simplify the notation.

	Every $t$-dependent vector field $X$ on $N$ gives rise to a family of standard vector fields on $N$ of the form  $\{X_t:x\in N\mapsto X(t,x) \in TN\}_{t \in \mathbb{R}}$. The {\it smallest Lie algebra} of $X$ is the smallest (in the sense of inclusion) Lie algebra of vector fields on $N$, let us say $V^X$, including the family of vectors $\{X_t\}_{t\in \mathbb{R}}$. If $V$ is a Lie algebra of vector fields on $N$, we write $\mathcal{D}^{V}$ for the distribution on $N$ spanned by the vector fields of $V$. Every  distribution $\mathcal{D}^{V}$ is regular in the connected components of a dense open subset of $N$ \cite{LS20,Va94}. In particular, if $V$ is finite-dimensional, then $\mathcal{D}^V$ is integrable (cf. \cite{La18,St80,Su73}).

	Each $t$-dependent vector field $X$ on a manifold $N$ can be considered as a vector field on $\mathbb{R}\times N$ via $\widetilde{\varphi}:T(\mathbb{R}\times N)\simeq T\mathbb{R}\times TN$. Moreover,  every $f\in C^\infty(\mathbb{R}
	\times N)$ can be considered as a $t$-parametrised family of functions $f_t\in C^\infty(N)$, with $t\in \mathbb{R}$, of the form $f_t:x\in N \mapsto f(t,x)\in\mathbb{R}$. Consequently, if $X$ is a vector field on $N$ and $f\in C^\infty(\mathbb{R}\times N)$,  we can understand $Xf$ as the function on $\mathbb{R}\times N$ such that $(Xf)_t=Xf_t$ for every $t\in \mathbb{R}$. Hence, a {\it $t$-dependent constant of motion} of $X$ is an  $f\in C^\infty(\mathbb{R}\times N)$ such that $Xf=0$.
	
	A {\it superposition rule} \cite{CGM07,Dissertationes,PW} for a system $X$ on a manifold $N$ is a map $\Psi:N^m\times N\rightarrow N$ satisfying that the general solution, $x(t)$, to $X$ can be written as
	$$
	x(t)=\Psi(x_{(1)}(t),\ldots ,x_{(m)}(t),k),
	$$
	for a generic family of particular solutions $x_{(1)}(t),\ldots,x_{(m)}(t)$  of $X$ and a parameter $k\in N$ to be related to the initial condition of $X$. We call {\it Lie system} a system of first-order ODEs admitting a superposition rule \cite{CGM07,LS,Dissertationes,PW}.
	
	\begin{theorem}{\bf (The Lie--Scheffers theorem \cite{CGM00,CGM07,LS,PW})} A system $X$ on $N$ admits a
		superposition rule if and only if 
		$
		X={\displaystyle \sum_{\alpha=1}^r}b_\alpha(t)X_\alpha
		$
		for a certain family $X_1,\ldots,X_r$ of vector fields on $N$ spanning an $r$-dimensional Lie algebra of vector fields, a so-called {\it Vessiot-Guldberg Lie algebra} of $X$, and
		a family $b_1(t),\ldots,b_r(t)$  of $t$-dependent functions.
	\end{theorem}
	
	One of the simplest non-trivial nonlinear examples of Lie systems is given by the Riccati equation \cite{LS}. Every Riccati equation is related to a $t$-dependent
	vector field on $\mathbb{R}$ of the form 
	$$
	X^{Ric}(t,x)=(a_1(t)+a_2(t)x+a_3(t)x^2)\frac{\partial}{\partial x},
	$$
	for certain $t$-dependent functions $a_1(t),a_2(t),$ and $a_3(t)$. We recall that it is sometimes assumed that $a_1(t)a_3(t)$ is not identically equal to zero. Then, $
	X^{Ric}=\displaystyle{\sum_{\alpha=1}^3} a_\alpha (t)X_\alpha,
	$
	where 
	$$
	X_1=\frac{\partial}{\partial x},\quad X_2=x\frac{\partial}{\partial x},\quad X_3=x^2\frac{\partial}{\partial x}
	$$ 
	are vector fields on $\mathbb{R}$ satisfying the commutation relations
	$$
	[X_1, X_2] = X_1,\qquad [X_1,X_3]=2X_2,\qquad [X_2,X_3]=X_3,
	$$
	and they therefore span a Lie algebra of vector fields isomorphic to $\mathfrak{sl}_2$ \cite{LS,PW}. According  to the Lie--Scheffers theorem, Riccati equations must admit a superposition rule. Indeed, it is known \cite{CGM00,In44,LS} that the general solution, $x(t)$, to a Riccati equation can be written in terms of a function $\Psi:\mathbb{R}^3\times \mathbb{R}\rightarrow \mathbb{R}$ in the form
	$$
	x(t)=\Psi(x_{(1)}(t),x_{(2)}(t),x_{(3)}(t),k),\qquad k\in \mathbb{R},
	$$
	where $x_{(1)}(t),x_{(2)}(t),x_{(3)}(t)$ are three different particular solutions to $X$ and 
	\begin{equation}\label{Eq:PSR}
	\Psi(u_1,u_2,u_3;k)=\frac{u_1(u_3-u_2)-ku_2(u_3-u_1)}{(u_3-u_2)-k(u_3-u_1)} ,
	\end{equation}
	and the limit $k\to \infty$ should be admitted to retrieve the particular solution $x_{(2)}(t)$. This latter remark about recovering $x_{(2)}(t)$ and the fact that (\ref{Eq:PSR}) is only well defined in an open dense subset of $\mathbb{R}^3\times\mathbb{R}$ explain why $\Psi$ is called, more properly, a {\it local superposition rule} \cite{Dissertationes}. We will not study this aspect in detail here as it is not relevant to our purposes and it is not important for practical applications (see \cite{Dissertationes,LS20} for details).
	
	Another relevant example of Lie system (see \cite{CGM00,Dissertationes,LS20}) is given by the system of first-order differential equations on an $r$-dimensional Lie group $G$ of the form
	\begin{equation}\label{Aut}
		X^G(t,g)=\sum_{\alpha=1}^rb_\alpha(t)X^R_\alpha(g),\qquad \forall g\in G,\quad\forall t\in \mathbb{R},
	\end{equation}
	where $X^R_1,\ldots,X_r^R$ stand for a basis of right-invariant vector fields
	on $G$ and $b_1(t),\ldots,b_r(t)$ are arbitrary $t$-dependent functions. Indeed,  if $R_g:h\in G\mapsto hg\in G$ is the right-translation map  and $\{e_1,\ldots,e_r\}$ 
	is a basis of $T_eG$, then the right-invariant vector fields $X^R_1,\ldots,X^R_r$, defined by  $X_\alpha ^R(g)=R_{g*e}e_\alpha $, span an $r$-dimensional Lie algebra of vector fields on $G$. Consequently,  the $t$-dependent vector field (\ref{Aut}) defines  a Lie system. The  Lie--Scheffers theorem states that  the differential equation  determining the integral curves of a $t$-dependent vector field $X^G$ on $G$ that takes the form  
	\begin{equation}\label{Aut2}\frac{dg }{dt}=X^G(t,g)\end{equation}  admits a superposition rule.
	A simple application of the right-translation $R_{g^{-1}*g}$ to both sides of (\ref{Aut2}) leads  to an  equivalent equation for the solutions $g(t)$, i.e.
	\begin{equation}\label{Aut3}R_{g^{-1}*g}\frac{dg}{dt} =  \sum_{\alpha=1}^rb_\alpha(t)e_\alpha\in T_eG.
	\end{equation}
	The right-invariance of the $t$-dependent vector field (\ref{Aut}) relative to the right action of $G$ on  itself, namely $R_{h*g}X^G(t,g)=X^G(t,hg)$ for every $g,h\in G$ and $t\in \mathbb{R}$, shows the right-invariance of equation (\ref{Aut2}), or its equivalent (\ref{Aut3}), i.e.  any particular solution $g_p(t)$ to (\ref{Aut2}) gives rise to a new particular solution $R_h g_p(t)$ of (\ref{Aut2})  for every $h\in G$.  As the initial conditions at $t=0$ determine univocally particular solutions of  (\ref{Aut2}), the general solution to  (\ref{Aut2}), let us say $g(t)$, can be
	brought into the form
	$$
	g(t)=R_hg_p(t),
	$$
	where $g_p(t)$ is any particular solution to (\ref{Aut2}) and $h\in G$. 
	Then, $X^G$ admits a superposition rule  involving one particular solution given by  $\Psi:(g,h)\in G\times G\mapsto R_hg\in G$.
	
	Lie systems of the form (\ref{Aut}) are called {\it automorphic Lie
		systems} \cite{Dissertationes}. Their special role in the theory of Lie systems is explained
	by the following theorem, which states that the general solution to every Lie system
	can be obtained from the knowledge of any particular solution of a related automorphic Lie system \cite{CGM00,Dissertationes,LS20,Ve93,Ve99}.
	
	\begin{theorem}
		Let $X$ be a Lie system on  $N$ of the form $X={\displaystyle\sum_{\alpha=1}^r}b_\alpha(t)X_\alpha$ for certain $t$-dependent functions $b_1(t),\ldots,b_r(t)$ and an $r$-dimensional Vessiot--Guldberg Lie algebra $V=\langle X_1,\ldots,X_r\rangle$. Let
		$G$ be the unique connected and simply connected Lie group whose Lie algebra is 
		isomorphic to $V$. Let $\varphi:G\times N\rightarrow N$ be the local Lie group action
		whose fundamental vector fields\footnote{We hereafter define the fundamental vector fields of a Lie group action $\varphi:G\times N\rightarrow N$ by $X_v(x)=\frac{d}{dt}\big|_{t=0}\varphi(\exp(-tv),x)$ for every $v\in T_eG$ and $x\in N$.} are spanned by $X_1,\ldots,X_r$. Then, the general  form of the integral curves,   $x(t)$, of $X$ can be written as $x(t)=\varphi(g(t),x_0)$, where $x_0\in N$ and $g(t)$ is any particular solution to (\ref{Aut2}) associated with the automorphic Lie system on $G$ of the form
		$
		X^G(t,g)=-{\displaystyle \sum_{\alpha=1}^r}b_\alpha (t)X_\alpha^R(g)$ for every $g\in G$ and $t\in \mathbb{R}$.

	\end{theorem}
	\section{On the definition of stratified Lie systems}\label{Se:SLS}
	
	Let us introduce our stratified Lie system notion and
	illustrate its usefulness with several examples of physical and mathematical interest. Our terminology slightly differs from the one in the previous literature, where stratified Lie systems are known as foliated Lie systems \cite{CGM00,LZ21} because of the reasons already given in the previous sections. 
	
	\begin{definition} We call a {\it stratified Lie system} on a manifold $N$ a $t$-dependent
		vector field  on $N$ of the form
		\begin{equation}\label{DFLS}
			X(t,x)=\sum_{\alpha=1}^rg_\alpha(t,x) X_\alpha(x),\qquad \forall t\in \mathbb{R},\quad \forall x\in N,
		\end{equation}
		where $X_1,\ldots,X_r$ span an $r$-dimensional  real  Lie algebra $V$ of vector fields  and $g_1,\ldots,g_r$ are common $t$-dependent constants of motion of the elements of $V$, i.e. $X_\alpha g_\beta =0$ on $N$ for every $\alpha,\beta=1,\ldots,r$. We call (\ref{DFLS}) and $V$ a {\it decomposition} and a {\it Vessiot--Guldberg Lie algebra} of the stratified Lie system $X$, respectively. If $X$ admits a decomposition (\ref{DFLS}) for a Vessiot-Guldberg Lie algebra $V$ so that the generalised distribution $\mathcal{D}^V$ is regular, we say that $X$ is a {\it foliated Lie system}.
	\end{definition}
	
	In virtue of the results by Stefan and Sussmann \cite{St80,La18,Su73}, the generalised distribution $\mathcal{D}^V$ associated with a Vessiot--Guldberg Lie algebra $V$ of a stratified Lie system is integrable and gives rise to a
	stratification $\mathcal{F}$ of $N$ such that the tangent  spaces to its strata  coincide  with $\mathcal{D}^V$. We call {\it $\mathcal{F}$-stratified Lie system} a stratified Lie system with a 
	Vessiot--Guldberg Lie algebra $V$ such that $\mathcal{D}^V$ consists of  the tangent spaces to the strata of $\mathcal{F}$. As the vector fields of $V$ are tangent to the
strata of the stratification $\mathcal{F}$, the system $X$ can be restricted to the strata of $\mathcal{F}$. Since $X_\beta g_\alpha =0$ for every $\alpha,\beta=1,\ldots,r$, the 
	restrictions of $g_1,\ldots,g_r$ to a stratum $\mathcal{F}_\lambda$ of $\mathcal{F}$ give rise to $r$ functions depending only on $t$. Indeed, consider a smooth curve\footnote{Smooth relative to the manifold structure with boundary on $[0,1]$ (see \cite{AM87} for details).} $\gamma:u\in [0,1]\subset\mathbb{R}\mapsto \gamma(u)\in\mathcal{F}_\lambda$ connecting two points of a stratum $\mathcal{F}_\lambda$. Then, the tangent vector at $\gamma(u)$ to the curve $\gamma$, let us say $\dot \gamma(u)$, can be written as a linear combination $\dot \gamma(u)=\displaystyle\sum_{\alpha=1}^rf_\alpha(u)X_{\alpha}(\gamma(u))$ of the values of the tangent vectors $X_1(\gamma(u)),\ldots,X_r(\gamma(u))$ spanning   $T_{\gamma(u)}\mathcal{F}_\lambda$ for certain $u$-dependent functions $f_1,\ldots,f_r:[0,1]\rightarrow \mathbb{R}$. Moreover,
	$$
	g(t,\gamma(1))-g(t,\gamma(0))=\int_0^1\frac{\partial }{\partial u }[g(t,\gamma(u))]du=\int_0^1\displaystyle\sum_{\alpha=1}^rf_\alpha(u) (X_\alpha g_t)(\gamma(u)) du=0.
	$$
	Consequently, $g(t,x)=g(t,x')$ for arbitrary points $x,x'\in \mathcal{F}_\lambda$ and any $t\in \mathbb{R}$. Therefore, the restriction of (\ref{DFLS}) to a stratum $\mathcal{F}_\lambda$ becomes a Lie system. More specifically, an $\mathcal{F}$-stratified Lie system gives rise to a Lie system on each stratum of $\mathcal{F}$ with a Vessiot--Guldberg Lie algebra  given by the restriction to the stratum of the vector fields of the Vessiot--Guldberg Lie algebra of the stratified Lie system. As a consequence of previous comments, the dimensions of all the induced  Vessiot--Guldberg Lie algebras on the strata are bounded. This last result can be used to show that defining a Lie system on each stratum of a stratification of a manifold does not necessarily give rise to a stratified Lie system. 
	
	In fact, for instance, the $t$-dependent vector field on $\mathbb{R}^3$ of the form 
	$$
	X(t,x,y,z)=\frac{\partial}{\partial x}+\sum_{n=0}^{\infty}e^{t(n+1)}f_n(z)x^n\frac{\partial}{\partial y},
	$$
	where it is assumed that $f_0(z)$ does never vanish and each $f_n(z)$, for $n\in \mathbb{N}$, is a smooth function that vanishes for every $z\in [-n,n]$ and it is different  from zero off this interval. Then, $X$ gives rise to a regular integrable generalised distribution $\mathcal{D}^{V^X}$ spanned by the vector fields
	$$
	\mathcal{D}^{V^X}_{(x,y,z)}=\left\langle \frac{\partial}{\partial x},\frac{\partial}{\partial y}\right\rangle, 
	$$
	with leaves given by the foliation $\mathcal{G}$ by planes in $\mathbb{R}^3$ of the form
	$$
	\mathcal{G}_z=\{(x,y,z)\in \mathbb{R}^3:x,y\in \mathbb{R}\},\qquad \forall z\in \mathbb{R}.
	$$
	The system $X$, at points with $k\leq |z_0|\leq k+1$, where $k\in \{0,1,2,3,\ldots\}$, takes the form
	$$
	X(t,x,y,z_0)=\frac{\partial}{\partial x}+\sum_{n=0}^{k}e^{t(n+1)}f_n( {z_0})x^n\frac{\partial}{\partial y},
	$$
	and, since the $t$-dependent vector field $X$ on $\mathbb{R}^3$ is tangent for every fixed $t$ to the submanifolds in $\mathbb{R}^3$ of the form $z=k$ for any $k\in \mathbb{R}$, the restriction of  $X$ to $\mathcal{G}_{z_0}$, let us say $X_{z_0}$, exists and it has a smallest Lie algebra $V^{z_0}_{k+1}=\{\partial/\partial x,\partial/\partial y,\ldots, x^k\partial/\partial y\}$. Then, all remaining Vessiot--Guldberg Lie algebras of the restriction of $X$  to $\mathcal{G}_{z_0}$ must contain $V^{z_0}_{k+1}$, which has dimension $k+1$. It turns out that the restrictions of $X$ to the leaves of $\mathcal{G}$ admit smallest Lie algebras whose dimensions cannot be bounded. Therefore, $X$ is not a stratified Lie system.
	
	Assume for a while that  $\mathcal{D}^V$ is not regular. As each finite-dimensional Lie algebra of vector fields on $N$ spans a generalised distribution whose rank is a lower semi-continuous function, then its rank is locally constant on an open dense subset of $N$ (see \cite{Va94} for details). Hence, every stratified Lie system is a foliated Lie system around a generic point.
	
	\section{Examples of foliated Lie systems}
	Let us provide several examples of foliated Lie systems
	with relevant physical applications. 
	
	As a first instance, given a Lie system defined by a $t$-dependent vector field on a manifold $N$ of the form
	\begin{equation}
		X(t,x)=\sum_{\alpha=1}^r b_\alpha(t)X_\alpha(x),\qquad \forall x\in N, \qquad \forall t\in \mathbb{R},
		\label{Liesyst}
	\end{equation}
	for arbitrary $t$-dependent functions $b_1(t),\ldots,b_r(t)$,
	so that there exist $r^3$ real numbers $c_{\alpha\beta}\,^\gamma$, with $\alpha,\beta,\gamma=1,\ldots,r$, such that $[X_\alpha, X_\beta]={\displaystyle\sum_{\gamma=1}^r}c_{\alpha\beta}\,^\gamma X_\gamma$. 
	If  $f\in C^\infty(N)$, then $ f\, X$ is not, in general,  a Lie system any more because the vector fields 
	$\bar X_\alpha=f\,X_\alpha$, with $\alpha=1,\ldots,r$, do not need to close on a finite-dimensional Lie algebra. However, if 
	$f \in C^\infty(N)$ is such that $X_\alpha f=0$ for $  \alpha=1,\ldots, r$, then the new
	$t$-dependent  vector field $f\, X$ is a  stratified Lie system. 
	
	Consider a $t$-dependent completely integrable Hamiltonian system $(h,\omega,T^*\mathbb{R}^n)$, where $h\in C^\infty(\mathbb{R}\times T^*\mathbb{R}^n)$ and $T^*\mathbb{R}^n$ is equipped with its canonical symplectic form $\omega_0$. Assume that there exists a $t$-dependent canonical transformation   $h(t,q_1,\ldots,q_n,p_1,\ldots,p_n)$, where $\{q_1,\ldots,q_n,p_1,\ldots,p_n\}$ stand for  some cotangent bundle coordinates for  $T^*\mathbb{R}^n$, onto a new $t$-dependent Hamiltonian $H(t,P_1(t),\ldots,P_n(t))$ that depends only on the momentum coordinates of a new  Darboux  coordinate system, $\{Q_1(t),\ldots,Q_n(t),P_1(t),\ldots,P_n(t)\}$, on $T^*\mathbb{R}^n$ for every $t\in \mathbb{R}$ (see \cite{AM87,Go80}).  In this case, the Hamilton equations for $H(t,P)$ read
	\begin{equation}\label{FLS1}
		\frac{dQ_i}{dt}=\frac{\partial H}{\partial P_i}(t,P),\qquad \frac{dP_i}{dt}=0,\qquad i=1,\ldots,n.
	\end{equation}
	This system is associated with the
	$t$-dependent vector field on $T^*\mathbb{R}^n$ given by
\begin{equation}\label{FLS1b}
	X^{HJ}(t,Q,P)=\displaystyle\sum_{i=1}^n\frac{\partial H}{\partial P_i}(t,P)\displaystyle\frac{\partial}{\partial Q_i}.
	\end{equation}
	The vector fields  $\{X^{HJ}_t\}_{t\in \mathbb{R}}$ span an Abelian Lie algebra $V^{X^{HJ}}$ of vector fields. If $
	V^{X^{HJ}}$ is finite-dimensional, then $X^{HJ}$ is a Lie system. Nevertheless, this does not need to be the case. For instance, if $H(t,P)={\displaystyle \sum_{i=1}^n}\cos(tP_i)$, then
	\begin{equation}\label{FLS}
		X^{HJ}(t,Q,P)=-\displaystyle\sum_{i=1}^n\sin(tP_i)t\frac{\partial}{\partial Q_i}
	\end{equation}
	and $V^{X^{HJ}}$ is the infinite-dimensional Abelian Lie algebra  spanned by the  vector fields 
	$
	Y_\lambda={\displaystyle\sum_{i=1}^n}\sin(\lambda P_i)\displaystyle\frac{\partial}{\partial Q_i},$ with $\lambda\in \mathbb{R}_+=\{x\in \mathbb{R}\mid x>0\}.$ Therefore, in this particular case, $X^{HJ}$ is not a Lie system.
	
	Independently of the specific form of $H(t,P)$, the manifold $T^*\mathbb{R}^n$ always admits a foliation $\mathcal{F}^{HJ}$ by leaves
	\begin{equation}\label{leavesTsQ}
		\mathcal{F}^{HJ}_k =\{(Q,P)\in T^*\mathbb{R}^n\mid P_1=k_1,\ldots,P_n=k_n\},
	\end{equation}
	parametrised via an $n$-tuple $k=(k_1,\ldots,k_n)\in \mathbb{R}^n$. 
	System (\ref{FLS1}) reduces on each $\mathcal{F}^{HJ}_k$ to 
	\begin{equation}\label{FLSk}
		\frac{dQ_i}{dt}=\frac{\partial H}{\partial P_i}(t,k),\qquad i=1,\ldots,n,
	\end{equation}
	which describes the integral curves of the restriction of (\ref{FLS1b}) to each $\mathcal{F}^{HJ}_k$, namely
	$$
	X^{HJ}_k=\sum_{i=1}^n\frac{\partial H}{\partial P_i}(t,k)\frac{\partial }{\partial Q_i}.
	$$
	The vector fields  $\{(X^{HJ}_k)_t\}_{t\in \mathbb{R}}$ on $\mathcal{F}^{HJ}_k$, for any fixed $k\in \mathbb{R}^n$, span a Lie algebra of vector fields,  $V_k$, contained in the finite-dimensional Lie algebra spanned by the restrictions to $\mathcal{F}_k^{HJ}$ of the vector fields on $T^*\mathbb{R}^n$ given by 
	\begin{equation}\label{VI}
		V^{HJ}=\left\langle \frac{\partial}{\partial Q_1},\ldots,\frac{\partial}{\partial Q_n} \right\rangle.
	\end{equation}
	Hence, $V_k$ is finite-dimensional and $X^{HJ}_k$ is a Lie system for every  $k\in \mathbb{R}^n$. Moreover, the vector fields in $V^{HJ}$ span  an $n$-dimensional  integrable regular distribution on $T^*\mathbb{R}^n$, whose leaves are given by (\ref{leavesTsQ}).
	Therefore, (\ref{FLS1}) becomes a foliated Lie system and $V^{HJ}$ is an associated Vessiot--Guldberg Lie algebra. 
	
	Every $X^{HJ}_k$ is related to a Lie system whose smallest Lie algebra is contained in the Lie algebra spanned by the restriction to $\mathcal{F}^{HJ}_k$ of the elements of $V^{HJ}$. Consequently, although the Lie systems $X^{HJ}_k$ may be different for distinct values of $k$, they  all are restrictions of a Lie system on $T^*\mathbb{R}^n$ with a Vessiot--Guldberg Lie algebra $V^{HJ}$. This relation will be studied more carefully in forthcoming sections.
	
	To illustrate more in detail the properties and applications of foliated and stratified Lie systems,
	let us consider another example: a $t$-dependent generalisation of a Lax pair studied in \cite{BV90} to analyse integrable systems. Let $\mathfrak{g}^{lp}=\mathbb{R}^n\rtimes \mathbb{R}^n$ 
 be the semi-direct sum of Lie algebras admitting a basis $\{e_i,h_i\}_{i=1,\ldots,n}$ such that $\mathbb{R}^n\simeq \langle e_1,\ldots, 
	e_n\rangle$, $\mathbb{R}^n\simeq \langle h_1,\ldots,h_n\rangle$ are Abelian Lie algebras and 
	\begin{equation}\label{AlgRel}
		[h_i,e_j]=2\delta^i_{j}e_j,\qquad i,j=1,\ldots,n,
	\end{equation}
	where $\delta^i_j$ is the Kronecker delta.
	Let $\{v^1,\ldots,v^{2n}\}$ be the basis of $\mathfrak{g}^{lp*}$   dual to the basis $\{v_1=e_1,\ldots,v_n=e_n,v_{n+1}=h_1,\ldots,v_{2n}=h_n\}$ of $\mathfrak{g}^{lp}$. 
	Hence, $\{v^1,\ldots,v^{2n}\}$ becomes a global coordinate system on $\mathfrak{g}^{lp}$. Define a family of $t$-dependent functions  $f_\alpha:  \mathbb{R}\times \mathfrak{g}^{lp}\mapsto \mathbb{R}$  of the form  $f_\alpha=f_\alpha(t,v^{n+1},\ldots,v^{2n})$ for $\alpha=1,\ldots,n$. i.e. the function really depends on the last $n$ last variables. 
	We set
	\begin{equation}\label{eq}
		\frac{dv}{dt}=-\sum_{\alpha=1}^{n}f_\alpha(t,v){\rm ad}_{v_\alpha}v=:X^{lp}(t,v),\qquad \forall v \in \mathfrak{g}^{lp},\quad \forall t\in \mathbb{R},
	\end{equation}
	where ${\rm ad}_{v_\alpha}v=[v_\alpha,v]$.
	
	If $\mathfrak{g}^{lp}$ is a matrix Lie algebra, the Lie bracket of $\mathfrak{g}^{lp}$ becomes the matrix commutator. Then,  (\ref{eq}) can be rewritten in the more common manner as a Lax pair
	\begin{equation}\label{FLS2}
		\frac{dv}{dt}=[v,m(t,v)],\qquad m(t,v)=\sum_{\alpha=1}^{n}f_\alpha(t,v) v_\alpha.
	\end{equation}
	If $\mathfrak{g}^{lp}$ is a general Lie algebra (not necessarily a matrix Lie algebra), one can alternatively extend by $C^\infty(\mathbb{R}\times \mathfrak{g}^{lp})$-linearity the Lie bracket in $\mathfrak{g}^{lp}$ to the space $C^\infty(\mathbb{R}\times \mathfrak{g}^{lp})\otimes \mathfrak{g}^{lp}$ of $\mathfrak{g}^{lp}$-valued $t$-dependent functions on $\mathfrak{g}^{lp}$. This allows us to use the expression (\ref{FLS2}) to describe every system (\ref{eq}).
	
	Consider the unique connected and simply connected Lie group $G^{LP}$ with Lie algebra $\mathfrak{g}^{lp}$. Then, $G^{LP}$ acts on $\mathfrak{g}^{lp}$ via the adjoint action ${\rm Ad}:(g,v)\in G^{LP}\times \mathfrak{g}^{lp}\mapsto {\rm Ad}_gv\in \mathfrak{g}^{lp}$. The fundamental vector fields of the adjoint action read
	\begin{equation}\label{FVFg}
		X^{\rm ad}_w(v)=\frac{d}{dt}\bigg|_{t=0}{\rm Ad}_{\exp(-tw)}v=-{\rm ad}_wv,
		\qquad \forall v,w\in \mathfrak{g}^{lp}.
	\end{equation}
	This enables us to bring (\ref{eq}) into the form 
	\begin{equation}\label{SF}
		\frac{dv}{dt}=\sum_{\alpha=1}^{n}f_\alpha (t,v)X^{\rm ad}_{v_\alpha}(v).
	\end{equation}
	In our chosen coordinate system and in view of (\ref{AlgRel}), the fundamental vector fields of the adjoint action for the Lie algebra $\mathfrak{g}^{lp}$ take the form 
	$$
	X^{\rm ad}_{e_\alpha}(v)=2v^{\alpha+n}\frac{\partial}{\partial v^{\alpha}},\qquad 
	X^{\rm ad}_{h_\alpha}(v)=-2v^{\alpha}\frac{\partial}{\partial v^{\alpha}},\qquad \alpha=1,\ldots,n.
	$$
	Hence, $X^{\rm ad}_{v_\alpha} f_\beta=0,$ for $\alpha=1,\ldots,2n$, $\beta=1,\ldots,n$. In particular, (\ref{eq}) takes the form
	$$
	\frac{dv^\alpha}{dt}=2f_\alpha(t,v)v^{\alpha+n},\qquad \frac{dv^{\alpha+n}}{dt}=0, \qquad \alpha=1,\ldots,n.
	$$
	Consider the Lie algebra $V^{\mathfrak{g}^{lp}}=\langle X^{\mathfrak{g}^{lp}}_1=2\partial/\partial v^{1},\ldots,X^{\mathfrak{g}^{lp}}_n=2\partial/\partial v^{n}\rangle$. The previous system can be rewritten as
	\begin{equation}\label{FLS3}
		\frac{dv}{dt}=\sum_{\alpha=1}^ng_\alpha(t,v)X^{\mathfrak{g}^{lp}}_\alpha(v),\qquad g_\alpha(t,v)=f_\alpha(t,v)v^{\alpha+n},\qquad \alpha=1,\ldots,n.
	\end{equation}
	The elements of $V^{\mathfrak{g}^{lp}}$ span a distribution $\mathcal{D}^{V^{\mathfrak{g}^{lp}}}$ of rank $n$ on $\mathfrak{g}^{lp}$. The leaves of the foliation, $\mathcal{F}^{\mathfrak{g}^{lp}}$, associated with $\mathcal{D}^{V^{\mathfrak{g}^{lp}}}$ on $\mathfrak{g}^{lp}$ take the form
	\begin{equation}\label{Leafg}
		\mathcal{F}^{\mathfrak{g}^{lp}}_k=\{ (v^1,\ldots,v^{2n})\in \mathfrak{g}^{lp}\mid v^{n+1}=k_{1},\ldots v^{2n}=k_{n}\},\qquad \forall k=(k_{1},\ldots,k_{n})\in \mathbb{R}^n.
	\end{equation}
	Moreover, the functions $g_
	\alpha$, with $\alpha=1,\ldots,n$, are $t$-dependent constants of motion of the vector fields belonging to $V^{\mathfrak{g}^{lp}}$. Therefore, (\ref{FLS3}) is a foliated Lie system associated with a Vessiot--Guldberg Lie algebra $V^{\mathfrak{g}^{lp}}$. In particular, the integration of the distribution $\mathcal{D}^{V^{\mathfrak{g}^{lp}}}$ gives rise to the foliation $\mathcal{F}^{\mathfrak{g}^{lp}}\!\!$. We can say then that $X$ is an $\mathcal{F}^{\mathfrak{g}^{lp}}\!\!$-foliated Lie system.
	
	The $t$-independent Lax pair studied in \cite{BV90} can be considered as a $t$-independent foliated Lie system of the form (\ref{FLS3}) with $f_\alpha=\partial h/\partial v^{\alpha+n}$, with $h=h(v^{n+1},\ldots,v^{2n})$ and $\alpha=1,\ldots,n$.
	
	System (\ref{FLS3}) does not need to be a Lie system. For example, one can consider the case when the $t$-dependent functions $g_\alpha$ on $\mathfrak{g}^{gl}$ take the form 
	$$
	g_\alpha(t,v)= \sin(t v^{\alpha+n})v^{\alpha+n},\qquad \alpha=1,\ldots,n.
	$$
Then,
	$$
	X^{lp}(t,v)= 2\sum_{\alpha=1}^n\sin(t v^{\alpha+n})v^{\alpha+n}\frac{\partial}{\partial v^\alpha}
	$$
	and $V^{X^{lp}}$ is infinite-dimensional.
	
	\section{On foliated superposition rules}\label{Sec:FLShat}
	Let us now study how one can obtain all the solutions to an $\mathcal{F}$-foliated Lie system passing through a leaf of $\mathcal{F}$ associated with the foliated Lie system in terms of other particular solutions passing through the same leaf. This will lead to introduce the notion of a foliated superposition rule.
	
	Consider again the foliated Lie system (\ref{FLS1}) on $T^*\mathbb{R}^n$. This system was associated with a foliation $\mathcal{F}^{HJ}$ whose leaves $\mathcal{F}^{HJ}_k$, with $k\in \mathbb{R}^n$, were given in (\ref{leavesTsQ}). Particular solutions to (\ref{FLS1}) have the form
	$$
	(Q^{(1)}(t),P_1=P_1^0,\ldots,P_n=P_n^0)
	$$
	for a point $P^0=(P^0_1,\ldots,P^0_n)\in \mathbb{R}^n$ and a particular solution, 
	$Q^{(1)}(t)$, to
	$$
	\frac{dQ_i}{dt}=\frac{\partial H}{\partial P_i}(t,P^0),\qquad i=1,\ldots,n.
	$$
	Moreover,  $(Q^{(1)}(t)+\hat{Q}, P^0_1,\ldots,P^0_n)$, for any $\hat {Q}\in \mathbb{R}^n$, is  another particular solution of (\ref{FLS1}) within $\mathcal{F}^{HJ}_k$. In fact, every solution to (\ref{FLS1}) within  
	$\mathcal{F}^{HJ}_k$ can be written as
	\begin{equation}
 		(Q(t),P^0_1,\ldots,P^0_n)=(Q^{(1)}(t)+\hat{Q},P^0_1,\ldots,P^0_n),
	\end{equation}
	for a unique $\hat Q\in \mathbb{R}^n$ and every expression of this latter form is a solution. This allows for defining a map $\Psi^{HJ}:T^*\mathbb{R}^n\times T^*\mathbb{R}^n\rightarrow T^*\mathbb{R}^n$ given by
	$$
	\Psi^{HJ}(Q^{(1)},P^{(1)};\hat{Q},\hat{P})=(Q^{(1)}+\hat{Q},\hat {P}),
	$$
	which satisfies that every solution, $\xi(t)=(Q(t),P(t))$, to (\ref{FLS1}) with initial condition in a $\mathcal{F}^{HJ}_k$ can be brought into the form
	$$
	\xi(t)=\Psi^{HJ}(\xi^{(1)}(t),\lambda),
	$$
	in terms of a particular solution $\xi^{(1)}(t)$ of (\ref{FLS1}) with an initial condition in $\mathcal{F}^{HJ}_k$ and a parameter $\lambda\in \mathcal{F}^{HJ}_k$. Moreover, there exists a one-to-one relation between the initial conditions of the solutions $\xi(t)$ of (\ref{FLS1}) in $\mathcal{F}^{HJ}_k$ and the values of $\lambda\in \mathcal{F}^{HJ}_k$. Finally, one has that $\Psi^{HJ}$ is a standard superposition rule involving one particular solution for any  Lie system on $T^*\mathbb{R}^n$ of the form
	$$
	\frac{dQ_i}{dt}=b_i(t),\qquad \frac{dP_i}{dt}=0,\qquad i=1,\ldots,n,
	$$
	for arbitrary $t$-dependent functions $b_1(t),\ldots,b_n(t)$. 
	
	One can see that the foliated Lie system related to the Lax pair (\ref{FLS2}) shares similar properties. This motivates to introduce 
	the following definition.
	
	\begin{definition} \label{Def:FLS}
		An {\it $\mathcal{F}$-foliated superposition rule} depending on $m$-particular solutions for a system $X$ on a manifold $N$ relative to a foliation $\mathcal{F}$ on $N$  is a superposition rule $\Psi: N^m \times N \rightarrow N$ for a certain Lie system on $N$ such that
		\begin{enumerate} 
			\item $\Psi(\mathcal{F}_k^{m+1})\subset \mathcal{F}_k$ for every leaf $\mathcal{F}_k$ of $\mathcal{F}$, 
			\item Every particular  solution, $x(t)$, of $X$
			 containing a point of a leaf $\mathcal{F}_k$ of $\mathcal{F}$, namely there exists $t_0\in \mathbb{R}$ such that $x(t_0)\in \mathcal{F}_k$, takes the form
			$$
			x(t)=\Psi(x_{(1)}(t),\ldots,x_{(m)}(t),\lambda)
			$$
			in terms of a particular set $x_{(1)}(t),\ldots, x_{(m)}(t)$ of $m$ generic particular solutions of $X$ contained in $\mathcal{F}_k$ and a unique $\lambda\in \mathcal{F}_k$. 
		\end{enumerate}
	\end{definition}
	
	Let us recall that it stems from the previous definition that if  $X$ admits an $\mathcal{F}$-foliated superposition rule, then the particular solutions of $X$ are always contained in a leaf of $\mathcal{F}$. Note also that the particular solutions of $\mathcal{F}$ taking values in each leaf of $\mathcal{F}$ are always described  in terms of the same number $m$ of particular solutions within the same leaf. We will refer to an $\mathcal{F}$-foliated superposition rule simply as a foliated superposition rule if $\mathcal{F}$ is known from the  context or its specific form is irrelevant to our considerations.  Note that if $X$ is restricted to a leaf $\mathcal{F}_k$ of $\mathcal{F}$, then all the solutions of $X$ within $\mathcal{F}_k$ can be written as a function of a generic family of particular solutions of $X$ within $\mathcal{F}_k$. Moreover, $X$ can be restricted to each $\mathcal{F}_k$ as every $X_t$, for $t\in \mathbb{R}$, is tangent to the leaves of $\mathcal{F}$. Additionally, the foliated superposition rule becomes a standard superposition rule for the restriction of $X$ to $\mathcal{F}_k$, which is a Lie system.

	In view of Definition \ref{Def:FLS}, the foliated Lie system (\ref{FLS1}) admits an $\mathcal{F}^{HJ}$-foliated superposition rule $\Psi^{HJ}:T^*\mathbb{R}^n\times T^*\mathbb{R}^n\rightarrow T^*\mathbb{R}^n$, where the leaves of $\mathcal{F}^{HJ}$ take the form (\ref{leavesTsQ}).

	Let us stress an interesting fact on the foliated superposition rule for (\ref{FLS1}).  The foliated Lie system (\ref{FLS1}) admits a Vessiot--Guldberg Lie algebra $V^{HJ}=\langle X_1=\partial/\partial Q_1,\ldots,X_n=\partial/\partial Q_n\rangle$. Then, $V^{HJ}$ is a Vessiot--Guldberg Lie algebra of the Lie systems on $T^*\mathbb{R}^n$ of the form
	$$
	\frac{dQ_i}{dt}=d_i(t),\qquad \frac{dP_i}{dt}=0,\qquad i=1,\ldots,n,
	$$
	for arbitrary $t$-dependent functions $d_1(t),\ldots, d_n(t)$. These Lie systems admit a superposition rule (cf. \cite{CGM07}) given by
	$$
	\begin{array}{rccc}
		\widetilde{\Psi}:&T^*\mathbb{R}^n\times T^*\mathbb{R}^n&\longrightarrow& T^*\mathbb{R}^n\\
		&({Q}^{(1)},{P}^{(1)};\hat{Q},\hat{P}))&\mapsto&(Q^{(1)}+\hat Q,\hat{P}).
	\end{array}
	$$
	Then, $\Psi^{HJ}$ becomes an $\mathcal{F}^{HJ}$-foliated superposition rule for (\ref{FLS1}). In other words, the foliated superposition rule for the foliated Lie system (\ref{FLS1}) can be retrieved via a superposition rule for the Lie systems related to its Vessiot--Guldberg Lie algebra.  An explanation of this fact will naturally appear in the next section. 
	
	Note that it is only the restrictions of a foliated superposition rule $\Psi:N^m\times N\rightarrow N$ to $\bigcup_{k}\mathcal{F}_k^{m+1}$ what really matters to the study of solutions of foliated Lie systems. Despite this fact, foliated superposition rules take the form $\Psi:N^m\times N\rightarrow N$ in applications (as in previous examples) and this latter form can easily be obtained in next sections from the relation between foliated superposition rules and standard superposition rules.  
	
	\section{On a foliated Lie--Scheffers theorem}
	
	Let us now  study the properties of first-order systems of ODEs admitting a foliated superposition rule. Our results can be
	considered as a generalisation of the standard geometric Lie--Scheffers theorem \cite{CGM07,LS,PW}. As a byproduct, our characterisation gives us an algorithm to obtain foliated superposition rules and clarifies the relation between superposition rules of Lie systems with a Vessiot--Guldberg Lie algebra $V$ and the foliated superposition rules for foliated Lie systems with the same Vessiot--Guldberg Lie algebra.
	
	Let us first recall a standard definition and a lemma in the previous literature on Lie systems (see \cite{CGM07,Dissertationes,LS20}).
	
	\begin{definition}  If $X$ is a vector field on $N$, let us say $X={\displaystyle\sum_{i=1}^n}X^i(x)\partial/\partial x^i$, its {\it diagonal prolongation} to $N^m$ is the vector field  on $N^m$ given by
		$$
		X^{[m]}=\sum_{a=1}^m\sum_{i=1}^{n}X^i(x_{(a)})\frac{\partial}{\partial x_{(a)}^i},
		$$
		where $\{x_{(a)}^i\}_{a=1,\ldots,m,i=1,\ldots,n}$ is the coordinate system on $N^m$ obtained by defining the same coordinate system $x^1\ldots,x^n$ on each copy of $N$ within $N^m$.
	\end{definition}
	\begin{lemma}\label{Lem:LS} (see \cite{CGM07,Dissertationes})
		Let $X_1,\ldots, X_r$ be vector fields on $N$ whose diagonal
		prolongations to $N^m$ are linearly independent at a generic point.
		If $b_1,\ldots,b_r\in C^\infty(N^{m+1})$, then ${\displaystyle\sum_{\alpha=1}^r}b_\alpha X^{[m+1]}_\alpha=Y^{[m+1]}$  for a vector field $Y$ on $N$ if and only if $b_1,\ldots,b_r$ are constant.
	\end{lemma}
	
	\begin{theorem}\label{FLST}
	  If $X$ is an $\mathcal{F}$-foliated Lie system, then it admits an $\mathcal{F}$-foliated superposition rule $\Psi:N^m\times N\rightarrow N$ such that  $
		ms\geq \dim V,
		$
		where $s$ is the dimension of the leaves of $\mathcal{F}$. 
	\end{theorem}
	\begin{proof}
Assume that $X$ is an $\mathcal{F}$-foliated Lie system admitting a decomposition (\ref{DFLS}) whose Vessiot--Guldberg Lie algebra $V$ is such that  $\mathcal{D}^V$ gives the tangent space to the foliation $\mathcal{F}$.  Let us construct an $\mathcal{F}$-foliated superposition rule for $X$.
		
		The $\mathcal{F}$-foliated Lie system $X$ gives rise on each leaf $\mathcal{F}_k$ of $\mathcal{F}$ to a Lie system 
		$$
		X_{k}(t,x)={\displaystyle\sum_{\alpha=1}^r}g_\alpha(t,k)X_\alpha|_{\mathcal{F}_k},
		$$
		where $X_\alpha|_{\mathcal{F}_k}$ is the restriction of $\mathcal{F}_k$ of $X_\alpha$, 
		that can be considered as the restriction to $\mathcal{F}_k$ of a Lie system on $N$ of the form $Y_{(k)}={\displaystyle\sum_{\alpha=1}^r}g_\alpha(t,k)X_\alpha$ possessing an $r$-dimensional Vessiot--Guldberg Lie algebra $V=\langle X_1,\ldots,X_r\rangle$. All Lie systems $Y_{(k)}$ admit a common superposition rule $\Psi:N^m\times N\rightarrow N$, but, in general,  $\Psi(\mathcal{F}_k^{m+1})\neq \mathcal{F}_k$. Hence, the curve $\Psi(x_{(1)}(t),\ldots,x_{(m)}(t),\lambda)$, where  $x_{(1)}(t),\ldots,x_{(m)}(t)$ are particular solutions to $X$ and $\lambda$ is within $\mathcal{F}_k$, is not in general a solution of $X_k$ and, therefore, is not a particular solution to $X$. Let us provide a method to obtain a superposition rule for all the $Y_{(k)}$ such that $\Psi(\mathcal{F}_k^{m+1})\subset \mathcal{F}_k$. This will give an $\mathcal{F}$-foliated superposition rule for $X$, because $\Psi(x_{(1)}(t),\ldots,x_{(m)}(t),\lambda)$ will be a particular solution to $Y_{(k)}$ within $\mathcal{F}_k$ and, therefore, a particular solution to $X$.

		On a neighbourhood of a generic point of $N$, there exists a local coordinate system adapted to the foliation $\mathcal{F}$ with $s$-dimensional leaves. In other words, we can construct a local coordinate system $\{\theta_1,\ldots,\theta_s,I_{s+1},\ldots, I_{n}\}$ whose first $s$ coordinates give rise to a coordinate system on each leaf $\mathcal{F}_k$ and the last $n-s$ coordinates are constant on each leaf of $\mathcal{F}$. Then, one can write in coordinates
		$$
		X_\alpha=\sum_{i=1}^sX^i_\alpha(\theta,I)\frac{\partial}{\partial \theta^i},\qquad \alpha=1,\ldots,r,\qquad \theta=(\theta^1,\ldots,\theta^s),\,\,I=(I^{s+1},\ldots,I^n),
		$$
		for certain functions $X^i_\alpha(\theta,I)$ for $i=1,\ldots,s$ and $\alpha=1,\ldots,r$. Let us restrict ourselves to a generic leaf 
		of $\mathcal{F}$, e.g. $$
		\mathcal{F}_k=\{(\theta_1,\ldots,\theta_s,I_{s+1},\ldots,I_{n})\in N:I_{s+1}=k_{s+1},\ldots,I_{n}=k_{n}\},\qquad k=(k_{s+1},\ldots,k_{n}).
		$$
		The vector fields $X_1,\ldots,X_r$ are tangent to $\mathcal{F}_k$. Hence, their restrictions to $\mathcal{F}_k$ can be considered as vector fields on the leaf. For $m$ large enough, the diagonal prolongations $X_1^{[m]},\ldots,X_r^{[m]}$ on $N^m$ become linearly independent at a generic point  (see \cite{Dissertationes} for a proof). As the vector fields $X_1^{[m]},\ldots,X_r^{[m]}$ are tangent to $\mathcal{F}_k^{m}$, one obtains
		\begin{equation}\label{FLC}
			ms\geq \dim V,
		\end{equation}
		where $s$ is by assumption the dimension of a leaf of $\mathcal{F}$. It is worth comparing  the above expression with Lie's condition, which only shows that $m \dim N
		\geq \dim V$.
		
		To obtain a foliated superposition rule for $X$, consider the diagonal prolongations $X_1^{[m+1]},\ldots,X_{r}^{[m+1]}$ to $N^{m+1}$. Let us define a local coordinate system on $N^{m+1}$ given by $\{\theta^{(a)}_1,\ldots,\theta^{(a)}_s,I^{(a)}_{s+1},\ldots, I^{(a)}_{n}\}$ with $a=0,\ldots, m$. The vector fields $X_1^{[m+1]},\ldots,X_{r}^{[m+1]}$ admit the common first-integrals $\Psi_{s+1}=I^{(0)}_{s+1},\ldots,\Psi_n=I^{(0)}_{n}$. Since $X^{[m+1]}_1,\ldots,X^{[m+1]}_r$ span a distribution of rank $\dim V\leq m\dim N$, one can find, at least, $s$ new functionally independent first-integrals, $\Psi_1,\ldots,\Psi_s$, common to $X_1^{[m+1]},\ldots,X^{[m+1]}_r$ such that 
		$$
		\frac{\partial(\Psi_1,\ldots,\Psi_n)}{\partial(\theta_1^{(0)},\ldots,\theta_s ^{(0)},I^{(0)}_{s+1},\ldots,I^{(0)}_{n})}\neq 0
		\Longrightarrow \frac{\partial(\Psi_1,\ldots,\Psi_s)}{\partial(\theta_1^{(0)},\ldots,\theta_s^{(0)})}\neq 0.
		$$ 
		This gives rise to a mapping $\widetilde{\Psi}:N^{m+1}\rightarrow N$ of the form
		$
		\widetilde{\Psi}(\theta^{(0)},I^{(0)},\ldots,\theta^{(m)},I^{(m)})=(\Psi_1,\ldots,\Psi_n).
		$
		Then, one can use the Implicit Function theorem to find the unique mapping $\Phi:N^m\times N\rightarrow N$ such that
		$$
		x_{(0)}=\Phi(x_{(1)},\ldots,x_{(m)},\sigma)\Leftrightarrow \widetilde{\Psi}(x_{(0)},\ldots,x_{(m)})=\sigma.
		$$
		In particular, 
		$
		\Phi(x_{(1)},\ldots,x_{(m)},\lambda)\in \mathcal{F}_{\bar k}
		$ for $\lambda=(\lambda_1,\ldots,\lambda_s,\bar{k})$, where $\lambda_1,\ldots,\lambda_s$ are arbitrary and $\bar k=(I_{s+1}^{(0)},\ldots, I_{n}^{(0)})$. In this way, given $m$ particular solutions $x_{(1)}(t),\ldots,x_{(m)}(t)$ to $X_{\bar{k}}$ belonging to the same leaf $\mathcal{F}_{\bar{k}}$, the mapping 
		$$
		\Phi(x_{(1)}(t),\ldots,x_{(m)}(t),k_1,\ldots,k_s,\bar k)=x(t)
		$$
		gives us every solution $x(t)$ to $X_{\bar{k}}$ within the leaf $\mathcal{F}_{\bar k}$ for every $(k_1,\ldots,k_s,\bar k)\in \mathcal{F}_{\bar k}$. Then, $x(t)$ is a solution of $X$ and the mapping $\Phi$ allows us to obtain all solutions to $X$ within $\mathcal{F}_{\bar{k}}$ out of particular solutions to $X$ in the same leaf and a parameter in $\mathcal{F}_{\bar k}$. 
		
	\end{proof}
	Observe that the last theorem  gives a procedure to construct a foliated superposition rule. This will be detailed and illustrated with examples in Section \ref{HOFSR}. 
	
	\begin{theorem}
		If a system $X$ on $N$
		admits an $\mathcal{F}$-foliated superposition rule $\Psi:N^m\times N\rightarrow N$, then there exists vector fields $X_1,\ldots,X_r$ on $N$ tangent to the leaves of $\mathcal{F}$ and common $t$-dependent constants of motion $f_1,\ldots,f_r\in C^\infty(\mathbb{R}\times N)$ for $X_1,\ldots,X_r$ so that
		$$
		X(t,x)=\sum_{\alpha=1}^rf_\alpha(t,x)X_\alpha(x)
		$$
		and
		\begin{equation}\label{Eq:rel}
			[X_\alpha,X_\beta]=\sum_{\gamma=1}^rh_{\alpha\beta}^\gamma X_{\gamma},\qquad \alpha,\beta=1,\ldots,r,
		\end{equation}
		where $h_{\alpha\beta}^\gamma$ are functions on $N$ taking constant values on the leaves of $\mathcal{F}$.
	\end{theorem}
	\begin{proof}
		Consider that we have an  $\mathcal{F}$-foliated superposition rule for a system $X$ on $N$ and the leaves of the foliation $\mathcal{F}$ have dimension $s$. By the definition of foliated superposition rules, the particular solutions to $X$ are contained in the leaves of $\mathcal{F}$ and the vector fields $X_t$ must be tangent to its leaves. 
		
		Let us fix a point $(x_{(1)},\ldots, x_{(m)})\in \mathcal{F}_k^m$ of a leaf $\mathcal{F}_k$ of $\mathcal{F}$. We can then use the Implicit Function theorem
		to obtain a new function $\widetilde{\Phi}:N^m\times N\rightarrow N$ such that
		$$
		F(x_{(1)},\ldots, x_{(m)},\lambda)=x_{(0)}\,\,\Longleftrightarrow \,\, \widetilde{\Phi}(x_{(0)},x_{(1)},\ldots, x_{(m)})=\lambda.
		$$
		The function $\widetilde{\Phi}$ is constant on generic families of $m+1$ particular solutions  of $X$, let say $x_{(0)}(t),\ldots, x_{(m)}(t)$,  belonging to the same leaf of the foliation $\mathcal{F}$. Therefore,
		$$       
		0=\frac{d}{dt}\widetilde{\Phi}(x_{(0)}(t),\ldots,x_{(m)}(t))=[X_t^{[m+1]}\widetilde{\Phi}](x_{(0)}(t),\ldots,x_{(m)}(t))=0.
		$$
		Recall that the foliated superposition rule induces on $\mathcal{F}^{m+1}_k$ a horizontal foliation over $\mathcal{F}^m_k$ relative to the projection onto the last $m$ copies of $\mathcal{F}_k$, i.e. the projection is a diffeomorphism between a leaf in $\mathcal{F}^{m+1}$ and its projection onto $\mathcal{F}^m$. The coordinates of the function $\widetilde{\Phi}$ take constant values exactly on the leaves of the horizontal foliation. Consequently, the vector fields $\{X^{[m+1]}_t\}_{t\in \mathbb{R}}$ span a distribution $\mathcal{D}_0$ contained in the tangent space to the leaves of the horizontal foliation on $\mathcal{F}^{m+1}_k$. We can extend such a distribution with the linear combinations of successive Lie brackets of $\{X^{[m+1]}_t\}_{t\in \mathbb{R}}$ to obtain a regular distribution $\mathcal{D}$, at a generic point of $\mathcal{F}_{\bar{k}}^{m+1}$, containing $\mathcal{D}_0$ and contained in the tangent space to the leaves of the foliation on $\mathcal{F}_{\bar k}^{m+1}$. 
		
		Consider a finite family of elements $X_1^{[m+1]},\ldots,X_u^{[m+1]}$ forming a local basis of the distribution $\mathcal{D}_0$. As the linear combinations of Lie brackets of diagonal prolongations are diagonal prolongations, the previous basis can be expanded to produce a family $X^{[m+1]}_1,\ldots,X^{[m+1]}_r$  of vector fields  spanning a regular distribution $\mathcal{D}$ almost everywhere. Since the leaves of our horizontal foliation on $\mathcal{F}^{m+1}_{\bar k}$ project diffeomorphically onto $\mathcal{F}^{m}_{\bar k}$, the vector fields $X^{[m]}_1,\ldots,X^{[m]}_r$ on $\mathcal{F}_{\bar k}^{m}$ become linearly independent at a generic point. By Lemma \ref{Lem:LS}, the Lie brackets of $X^{[m+1]}_1,\ldots,X^{[m+1]}_r$ close a finite-dimensional Lie algebra of vector fields. Let us denote by $_kX_1,\ldots,_kX_r$ the restrictions of $X_1,\ldots, X_r$ to $\mathcal{F}_k$. Hence, 
		$$
		[\,_{\bar k}X_\alpha,\,_{\bar k}X_\beta]=\sum_{\gamma=1}^rc_{\alpha\beta}\,^\gamma(\bar{k}) \,_{\bar k}X_\gamma
		$$
		on $\mathcal{F}_{\bar k}$. 
		The previous procedure can be extended to other leaves $\mathcal{F}_{k'}$ for different values $k'$ close enough to $k$ so that the vector fields $X^{[m+1]}_1,\ldots,X^{[m+1]}_r$ for the initial $\mathcal{F}_k$ can also be  used locally. 
		Moreover, $X_t^{[m+1]}$ must be on each leaf a linear combination of the $X_1^{[m+1]},\ldots,X_r^{[m+1]}$ with coefficients given by $t$-dependent functions on $N^{m+1}$. Hence, 
		$$
		X_t^{[m+1]}(\xi)=\sum_{\alpha=1}^{r}g_\alpha(t,\xi) X_\alpha^{[m+1]}(\xi),
		$$
		for certain functions $g_\alpha(t,\xi)$.
		Let us restrict the above expression to an arbitrary leaf $\mathcal{F}^{m+1}_k$.  Using Lemma \ref{Lem:LS}, we obtain on $\mathcal{F}_k$ that 
		$$
		X_t(x)=\sum_{\alpha=1}^{r}g_\alpha(t,x) X_\alpha(x),
		$$
		and $g_\alpha(t,x)=g_\alpha(t,x')$ for points $x,x'$ in the same leaf of $\mathcal{F}$ and every $t\in \mathbb{R}$. 
	\end{proof}
	
	The proof of the last theorem almost reassembles the proof of the Lie--Scheffers theorem. In despite of that, there exists a relevant difference. At the very end, we cannot ensure that we obtain a foliated Lie system. The problem is that the vector fields $X_1,\ldots,X_r$, which are tangent to the leaves of $\mathcal{F}$, may close a different Lie algebra on each leaf (see  the example in Section \ref{Se:SLS}). Therefore, there will be no vector fields $\hat{X}_\alpha$ such that $X_\alpha={\displaystyle \sum_{\beta=1}^r}f_{\alpha\beta} \hat{X}_\beta$  for functions $f_{\alpha\beta}$ that are constants on the leaves of $\mathcal{F}$ in such a way that $\langle\hat{X}_1,\ldots,\hat{X}_r\rangle$ is a Lie algebra of vector fields on $N$. This is due to the fact that if such vector fields exist, then the vector fields $X_1,\ldots,X_r$ must span a Lie subalgebra of the one spanned by $\hat X_1,\ldots,\hat X_r$.


	\section{Automorphic foliated Lie systems}\label{Sec:AFLS}
	We showed in Section \ref{Intro} that the evolution of a Lie system can be determined by an automorphic Lie system. More generally, we are going to prove that the evolution of a foliated Lie system
	can be obtained via a foliated Lie system on a principal
	bundle of a particular type.
	\begin{definition}
		Consider a trivial principal bundle $\pi: G\times M\rightarrow M$ with structural $r$-dimensional Lie group $G$ acting on $G\times M$ by $\varphi_\pi:(g_1,(g_2,k))\in G\times (G\times M)\mapsto (g_1g_2,k)\in G\times M$ and a $t$-dependent
		vector field  on $G\times M$ given by
		\begin{equation}
			\label{PAut}
			X_T^R(t,g,k) =\sum_{\alpha=1}^{r}f_\alpha(t,k) X^R_\alpha(g,k),
		\end{equation}
		where  $X^R_1,\ldots,X^R_r$ form a basis of fundamental vector fields of the action of $G$ on $G\times M$, and $f_1,\ldots,f_r\in C^\infty(\mathbb{R}\times G\times M)$ are $t$-dependent common constants of motion of $X^R_1,\ldots,X_r^R$, which therefore can be considered as functions on $\mathbb{R}\times M$. We call (\ref{PAut}) an {\it automorphic foliated Lie system}.
	\end{definition}
	
Note that every vector field on $G$ can be considered in a natural way as a vector field on $G\times M$ via the vector bundle isomorphism $T(G\times M)=TG\times TM$. Indeed, they can be considered as vertical vector fields relative to the projection $\pi:G\times M\rightarrow M$. Then, $X^R_1,\ldots,X^R_r$ can be understood as right-invariant vector fields  on $G$. They also span a finite-dimensional Lie algebra of vector fields on $G\times M$ and a regular $r$-dimensional generalised distribution. In fact, at each point of the bundle the values of $X^R_1,\ldots,X_r^R$  span the vertical space at such point of the  bundle $\pi:G\times M\rightarrow M$.
	Since $f_1,\ldots,f_r$ are common $t$-dependent constants of motion for vertical vector fields, they are constant on the fibres of $\pi:G\times M\rightarrow M$ and they can be considered, in a unique manner, as $t$-dependent functions on $M$. Therefore,	(\ref{PAut}) is a foliated Lie system. 
	
	When $M$ is a point,  $f_1,\ldots,f_r$ become functions depending only on $t$ and (\ref{PAut}) turns into a standard automorphic Lie system \cite{Dissertationes}. More generally, every trivial principal bundle $\pi:P\rightarrow M$ with a structural group $G$ and fundamental vector fields spanned by $X_1,\ldots,X_r$ gives rise to a $t$-dependent vector field $X$ on $P$ of the form
	$$
	X(t,x)=\sum_{\alpha=1}^rf_\alpha(t,x)X_\alpha(x),\qquad \forall x\in P,\quad \forall t\in \mathbb{R},
	$$
	where $f_1,\ldots,f_r\in C^\infty(\mathbb{R}\times P)$ are $t$-dependent constants of motion of $X_1,\ldots,X_r$. The trivialisation of the principal bundle $\pi:P\rightarrow M$ maps diffeomorphically $X$ onto  $X^R_\pi$. Consequently, $X$ is, up to a local bundle diffeomorphism, an automorphic foliated Lie system.
	
	Let us prove that, in analogy with Lie systems, every $\mathcal{F}$-foliated Lie system gives rise to an automorphic foliated Lie system, whose solutions allow us to obtain the general solution of the foliated Lie system. Recall that every $t$-dependent function on a fiber bundle that is a $t$-dependent constant of motion of its vertical vector fields can be considered as a $t$-dependent function on the base manifold in a canonical way.

	\begin{theorem}\label{FLSAFLS} Let $X(t,x)={\displaystyle  \sum_{\alpha=1}^{r}}f_\alpha(t,x) X_\alpha(x)$ be a foliated Lie system on $N$ associated with an $r$-dimensional Vessiot--Guldberg Lie algebra $V=\langle X_1,\ldots,X_r\rangle$. Let $\phi:(g,x)\in G\times N\mapsto \phi_x(g)=\phi(g,x)\in N$ be the effective local Lie group action associated with the integration of the vector fields of  $V$. Assume that the space $M$ of leaves of $\mathcal{D}^V$ admits a manifold structure so that the  $\pi:N\rightarrow N/G=M$. Let us  define the automorphic foliated Lie system on the total space of the principal bundle $\pi:G\times M\rightarrow M$  given by
		\begin{equation}\label{RAFLS}
			X_\pi^R(t,g,k)=-\sum_{\alpha=1}^rf_\alpha(t,k)X^R_\alpha(g,k),\qquad \forall t\in \mathbb{R},\quad \forall g\in G,\quad \forall k\in M,
		\end{equation}
		where  $-X^R_\alpha$ and $X_\alpha$ are the fundamental vector fields associated with the same element of $T_eG$ relative to the actions $\varphi_\pi$ and $\phi$, correspondingly. Then, each particular solution $x(t)$ of $X$ contained in the leaf indexed by $k\in M$ can be written as
		\begin{equation}\label{Rel}
			x(t)=\phi(g(t),x(0)),
		\end{equation}
		where $\gamma:t\in \mathbb{R} \mapsto (g(t),k)\in G\times M$ is a particular solution to $X_\pi^R(t,g,k)$ with $g(0)=e$.
	\end{theorem}
	\begin{proof} Let us prove that $x(t)$ given by (\ref{Rel}) is a particular solution to $X$ for every $x(0)$. Using (\ref{Rel}), we have that
		\begin{equation}\label{Exp1}
			\frac{dx}{dt}(t_0)=\frac{d}{dt}\bigg|_{t=t_0}\phi(g(t)g^{-1}(t_0),\phi(g(t_0),x(0)))=T_e\phi_{\phi(g(t_0),x(t_0))}\left(\frac{d}{dt}\bigg|_{t=t_0}g(t)g^{-1}(t_0)\right).
		\end{equation}
		Since $(g(t),x)$ is a particular solution to (\ref{RAFLS}) and $X^R(g)=R_{g*e}X^R(e)$ for every $g\in G$, one has  that
		$$
		\frac{dg}{dt}(t_0)=-\sum_{\alpha=1}^rf_\alpha(t_0,k)R_{g(t_0)*e}X_\alpha^R(e)\Rightarrow 	R_{g^{-1}(t_0)*g(t_0)}\frac{dg}{dt}(t_0)=-\sum_{\alpha=1}^rf_\alpha(t_0,k)X^R_\alpha(e).
		$$
		Since $-X_\alpha^R$ and $X_\alpha$ are fundamental vector fields related to the same element of $T_eG$ relative to the $\phi$ and $\varphi_\pi$ actions, one obtains, after considering $X_\alpha^R$ as a right-invariant vector field in the natural way, that $-T_e\phi_{x}X_\alpha^R(e)=X_\alpha(x)$ for every $x\in N$ and $\alpha=1,\ldots,r$. Substituting this relation in (\ref{Exp1}) and since the functions $f_\alpha(t,x)$ are just $t$-dependent on the orbits under the action of $\phi$, we obtain
		$$
		\frac{dx}{dt}(t_0)=-T_e\phi_{\phi(g(t_0),x(t_0))}\left[\sum_{\alpha=1}^rf_\alpha(t_0,k)X^R_\alpha(e)\right]=\sum_{\alpha=1}^rf_\alpha(t_0,x(t_0))X_\alpha(x(t_0)).
		$$
		In fact, it worth noting that $f_\alpha(t,x(t_0))$ depends only on $t$ and $\pi(x(t_0))$.
		Therefore, $\phi(g(t),x(0))$ is a solution of $X$ for every $x(0)\in N$. Note that every solution of $X$ on $N$ is of the above form and the result of our theorem follows.
	\end{proof}
	
	\section{Applications}
	Let us use our results of previous sections to study several physical problems. First,  we focus our attention on an extension of the generalised Ermakov system \cite{Le91}. Next, we study Lax pairs for $t$-dependent Hamiltonian systems, extending and explaining certain results in \cite{BV90}. In this case, we obtain the automorphic foliated Lie system associated with it. We also analyse the existence of related Hamiltonian structures, which extends, in a geometric way, some of the results given in \cite{BV90}.

	\subsection{A new class of generalised Ermakov systems}
	
	There exists an extensive literature on the so-called Ermakov systems and their generalisations (see \cite{CLR08c,Go90,Le91,Ra80,RR79,RR80} and references therein). We propose here a class of generalised Ermakov systems that cannot be described through Lie systems but they admit a description in terms of foliated ones. 
	
	Ray and Reid introduced the so-called {\it generalised Ermakov systems} \cite{RR79,RR80}  on $\mathbb{R}_0^2=\{(x,y):xy\neq 0\}$ of the form
	\begin{equation}\label{GEE}
		\ddot x=-\omega^2(t)x+\frac{g(y/x)}{yx^2},\qquad \ddot y=-\omega^2(t)y+\frac{f(x/y)}{xy^2},\qquad (x,y)\in \mathbb{R}^2_0,\quad t \in \mathbb{R},
	\end{equation}
	where $f$ and $g$ are real functions $f,g:\mathbb{R}\rightarrow \mathbb{R}$ and the dots over the variables $x,y$ stand for their time  derivatives. The equation introduced by Milne \cite{M30} was given by one of the above equations, let us say the one depending on $y$ with $f(x/y)=x/y$,  and its mathematical study can be found in \cite{P50}. Meanwhile, the so-called  Ermakov--Pinney  system  	corresponds to the particularisation $g(u)=u$ and $f(u)=0$ for every $u\in \mathbb{R}$. This generalisation admits a constant of motion, the so-called {\it generalised Lewis invariant} (see \cite{Dissertationes} and references therein), of the form
\begin{equation}\label{defI}
	I(x,y,\dot x,\dot y)=\frac 12\left(x\dot y-y\dot x\right)^2+\int^{x/y}[f(u)-u^{-2}g(1/u)]du.
	\end{equation}
	It was noted by Ray, Reid, and Goedert \cite{Go90,Ra80,RR80}, that the term $\omega^2(t)$ can be replaced by much more general expressions (cf. \cite{Le91}). For instance, there exist generalisations of (\ref{GEE}) where $\omega(t)$ depends on the time-derivatives of $x$ and $y$ \cite{RR80}.  As a new generalisation, we propose the second-order system of differential equations given by
	\begin{equation}\label{Leach}
		\ddot x={-\omega^2(t,I)x}+\frac{g(I,y/x)}{x^2y },\quad \ddot y={-\omega^2(t,I)y}+\frac{f(I,x/y)}{xy^2 },\quad (x,y)\in \mathbb{R}_0^2,\quad t\in\mathbb{R},
	\end{equation}
	where $f,g:\mathbb{R}\rightarrow \mathbb{R}$ are arbitrary non-vanishing functions and $I$ is given by (\ref{defI}). In the case where $f(I,u)=u$ and $g(I,u)=u$ for every $u\in \mathbb{R}$, one recovers the generalised Ermakov system studied in \cite{Le91}. 
	Consider the system of first-order differential equations on $T\mathbb{R}_0^2$ of the form
$$
	\left\{
	\begin{array}{rcl}
	\dfrac{dx}{dt}&=&v_x,\\ \dfrac{dy}{dt}&=&v_y,\\
	\dfrac{dv_x}{dt}&=&{-\omega^2(t,I)x}+\dfrac{g(I,y/x)}{x^2y },\\
  \dfrac{dv_y}{dt}&=&{-\omega^2(t,I)y}+\dfrac{f(I,x/y)}{y^2x },
  \end{array}\right.
	$$
	obtained by adding the variables $v_x=\dot x$ and $v_y=\dot y$ to (\ref{Leach}) and where $I$ is a function as in (\ref{defI}) but with $\dot x$ and $\dot y$ replaced by $v_x$ and $v_y$, respectively. This system is associated with the $t$-dependent vector field $X=\omega^2(t,I)X_3+X_1$ on $T\mathbb{R}_0^2$, where the vector fields 
$$
	\begin{array}{rcl}
	X_1&=&\dfrac{f(I,x/y)}{xy^2}\dfrac{\partial}{\partial v_y}+v_y\dfrac{\partial}{\partial y}+\dfrac{g(I,y/x)}{x^2y }\dfrac{\partial}{\partial v_x}+v_x\dfrac{\partial}{\partial x},\\
	X_2&=&\dfrac 12\left[y\dfrac{\partial}{\partial y}-v_y\dfrac{\partial}{\partial v_y}+x\dfrac{\partial}{\partial x}-v_x\dfrac{\partial}{\partial v_x}\right],\\
	X_3&=&-y\dfrac{\partial}{\partial v_y}-x\dfrac{\partial}{\partial v_x},\end{array}
	$$
	have the commutation relations
	$$
	[X_1,X_2]=X_1,\qquad [X_1,X_3]=2X_2,\qquad [X_2,X_3]=X_3{,}
	$$
	and therefore span a Lie algebra of vector fields $V^{gES}$ isomorphic to $\mathfrak{sl}(2,\mathbb{R})$. A straightforward calculation shows that $I$ is a common first-integral to $X_1,X_2,X_3$ and $X_1\wedge X_2\wedge X_3\neq 0$ on a  dense open subset  $\mathcal{O}\subset T\mathbb{R}_0^2$. Then, $X$ becomes a foliated Lie system on $\mathcal{O}\subset T\mathbb{R}^2_0$. 
	
	Let us comment on fact that (\ref{Leach}) admits a Lie algebra of Lie symmetries isomorphic to $\mathfrak{sl}(2,\mathbb{R})$ \cite{Le91}.  Indeed, note that the restriction of $X$ to every leaf of the foliation determined by the integral leaves of the distribution with $\mathcal{D}^{V^{gES}}$ becomes a Lie system. Since $X_1\wedge X_2\wedge X_3\neq 0$ on these leaves, one obtains a so-called {\it locally automorphic} Lie system on each leaf (see \cite{GMLV19}). It was proved in \cite{GMLV19} that such a Lie system admits a Lie algebra of Lie symmetries isomorphic to $\mathfrak{sl}(2,\mathbb{R})$. Gluing together these vector fields on each leaf, we obtain a Lie algebra of Lie symmetries of $X$ on $T\mathbb{R}_0^2$ isomorphic to $\mathfrak{sl}(2,\mathbb{R})$. This feature is common to many other generalisations of Ermakov systems \cite{Le91}.
	
	%
	
	%
	
	\subsection{How to obtain foliated superposition rules}\label{HOFSR}
	The proof of Theorem \ref{FLST} shows in an implicit way how to obtain a foliated superposition rule for a foliated Lie system. This section aims to illustrate this method  and to describe it  in detail. A careful reading of the proof of Theorem \ref{FLST} shows that the steps of the method go as follows:
	
	\begin{itemize}
		\item Consider a Vessiot--Guldberg Lie algebra $V_{\mathcal{F}}$ of an  $\mathcal{F}$-foliated Lie system $X$ on an $d_T$-dimensional manifold $T$. 
		\item Find the smallest $m\in \mathbb{N}$ so that the diagonal prolongations of the vector fields of $V_{\mathcal{F}}$ span a distribution of rank $\dim V_{\mathcal{F}}$ at a generic point of $T^m$.  This states the number of particular solutions of the foliated superposition rule, namely $m$.
		\item Consider a coordinate system $\theta^1,\ldots,\theta^s,I^{s+1},\ldots, I^{d_T}$, adapted to the foliation $\mathcal{F}$ around a generic point $x\in T$, i.e. the first $s$ coordinates give rise to a local coordinate system on each leaf $\mathcal{F}_k$ of $\mathcal{F}$, while the last ${d_T}-s$ coordinates are constant on the leaves of $\mathcal{F}$. Define the same coordinate system $\theta^1,\ldots,\theta^s,I^{s+1},\ldots, I^{d_T}$ on each copy $T$ within $T^{m+1}$. This gives rise to a coordinate system on $T^{m+1}$ of the form $\theta_{(a)}^1,\ldots,\theta_{(a)}^s,I_{(a)}^{s+1},\ldots, I_{(a)}^{d_T}$ for $a=0,\ldots,m$. 
		\item Define $\Psi^{s+1}=I_{(0)}^{s+1},\ldots,\Psi^{{d_T}}= I_{(0)}^{d_T}$  and obtain $s$ common first-integrals $\Psi^1,\ldots, \Psi^s$ for the diagonal prolongations of the elements of $V_{\mathcal{F}}$ to $T^{m+1}$ satisfying that
		$$
		\frac{\partial(\Psi^1,\ldots,\Psi^{d_T})}{\partial (\theta_{(0)}^1,\ldots,\theta_{(0)}^s,I_{(0)}^{s+1},\ldots, I_{(0)}^{d_T})}\neq 0\Leftrightarrow \frac{\partial(\Psi^1,\ldots,\Psi^s)}{\partial (\theta_{(0)}^1,\ldots,\theta_{(0)}^s)}\neq 0.
		$$
		\item Assume $\Psi^1=k_1,\ldots, \Psi^{d_T}=k_{d_T}$ and obtain $\theta^1_{(0)},\ldots,\theta^s_{(0)},I_{(0)}^{s+1},\ldots,I_{(0)}^{d_T}$ as a function of $k_1,\ldots,k_{d_T}$ and $\theta_{(a)}^1,\ldots,\theta_{(a)}^s,I_{(a)}^{s+1},\ldots, I_{(a)}^{d_T}$ for $a=1,\ldots,m$, i.e.
		$$
		\theta^i_{(0)}=F^i(\theta_{(1)},I_{(1)},\ldots,\theta_{(m)},I_{(m)},k_1,\ldots, k_{d_T}),\qquad
		I^j_{(0)}=k_j,
		$$
		for certain functions $F^i:T^m\times T\rightarrow \mathbb{R}$, with $i=1,\ldots,s$ and $j=s+1,\ldots,{d_T}$.  This gives rise to a superposition rule $\Phi:T^m\times T\rightarrow T$ for every Lie system with a Vessiot--Guldberg Lie algebra $V_{\mathcal{F}}$ of the form (see {\cite{CGM07,Dissertationes,LS20}} for details)
		$$
		\Phi(\theta_{(1)},I_{(1)},\ldots,\theta_{(m)},I_{(m)},k_1,\ldots,k_{d_T})=(F^1,\ldots,F^s,I^{s+1}_{(0)},\ldots,I^{d_T}_{(0)}).
		$$
		Hence, the map $\Phi:T^m\times T\rightarrow T$ becomes an $\mathcal{F}$-foliated superposition rule for $X$ at a generic point of $T$.
	\end{itemize}
	
	Let us illustrate our method by applying it to the Lax pair (\ref{FLS3}). In this case, the manifold where the Lax pair is defined is $2n$-dimensional and $s=n$. We recall that the foliated Lie system given by the Lax pair  (\ref{FLS3}) was related to the Vessiot--Guldberg Lie algebra $V^{\mathfrak{g}^{lp}}$. The vector fields of a basis of $V^{\mathfrak{g}^{lp}}$ are linearly independent at a generic point. Hence, we can obtain a foliated superposition rule depending on one particular solution. 
	The coordinates $\theta^1=v^1,\ldots,\theta^n=v^n,I^1=v^{n+1},\ldots,I^n=v^{2n}$ are adapted to the foliation $\mathcal{F}^{\mathfrak{g}^{lp}}$ of the system under study.  Consider the coordinate system on $(\mathfrak{g}^{lp})^2$ of the form $\theta_{(a)}^1,\ldots,\theta_{(a)}^n,I_{(a)}^1,\ldots,I_{(a)}^n$ with $a=0,1$. Take the diagonal prolongations of the elements of $V^{\mathfrak{g}^{lp}}$ to $(\mathfrak{g}^{lp})^2$. To obtain $2n$ functionally independent constants of motion for such diagonal prolongations choose  $\Psi^{s+1}=I^{1}_{(0)},\ldots,\Psi^{n}=I^n_{(0)}$ and  $\Psi^i=\theta^i_{(0)}-\theta^i_{(1)}$ for $i=1,\ldots,n$. Then,
	$$
	\frac{\partial (\Psi^1,\ldots,\Psi^n)}{\partial( \theta^1_{(0)},\ldots, \theta^n_{(0)},I^1_{(0)},\ldots I^{n}_{(0)})}\neq 0.
	$$
	By fixing $\Psi^1=k_1,\ldots,\Psi^{2n}=k_{2n}$, a superposition rule $\Phi: \mathfrak{g}^{lp}\times  \mathfrak{g}^{lp}\rightarrow \mathfrak{g}^{lp}$ for every Lie system with a Vessiot--Guldberg Lie algebra $V^{\mathfrak{g}^{lp}}$ reads
	$$
	\Phi(\theta_{(1)},I_{(1)},k)=(\theta^1_{(1)}+k_1,\ldots,\theta^n_{(1)}+k_n,k_{n+1},\ldots,k_{2n}).
	$$
	Restricting oneself to the case  $k_{n+1}=I^1_{(1)},\ldots,k_{2n}=I^n_{(1)}$, one gets an $\mathcal{F}^{\mathfrak{g}^{lp}}$-foliated superposition rule $\Psi^\mathfrak{g}:\mathfrak{g}^{lp}\times \mathfrak{g}^{lp}\rightarrow\mathfrak{g}^{lp}$ such that the particular solutions to $X^{lp}$ in the leaf $\mathcal{F}_k^{\mathfrak{g}^{lp}}$, with $k=I_{(1)}\in \mathbb{R}^n$, are of the form
	$$
	\Psi(\theta_{(1)}(t),I_{(1)},k_1,\ldots,k_{2n})=(\theta^1_{(1)}+k_1,\ldots,\theta^n_{(1)}+k_n,I^1_{(1)},\ldots,I^n_{(1)})
	$$
	for a particular solution $(\theta_{(1)}(t),I_{(1)})$ of $X^{lp}$ in $\mathcal{F}_k^{\mathfrak{g}^{lp}}$.
	\subsection{Lax pairs and automorphic Lie systems}
	Let us study the systems (\ref{FLS1}) and (\ref{FLS3}) by means of a common automorphic foliated Lie system. 
	
	The foliated Lie system (\ref{FLS1}) is associated with a Vessiot--Guldberg Lie algebra $V^{HJ}$ of the form (\ref{VI}), which is isomorphic to the Lie algebra $(\mathbb{R}^n,+)$. We denote by $\{\lambda_1,\ldots,\lambda_n\}$ the dual basis to the canonical basis $\{e_1,\ldots,e_n\}$ on $\mathbb{R}^n$.
	The Lie group action obtained by integrating the vector fields of $V^{HJ}$ reads
	$$
	\begin{array}{rccc}
		\varphi:&\mathbb{R}^{n}\times T^*\mathbb{R}^n&\rightarrow& T^*\mathbb{R}^n,\\
		&(\lambda,Q,P)&\mapsto&(Q-\lambda,P),
	\end{array}
	$$
	where we denote $Q=(Q_1,\ldots,Q_n)$, $\lambda=(\lambda_1,\ldots,\lambda_n)$, and $P=(P_1,\ldots,P_n)$. 
	Observe that the Lie group action has been chosen so that the fundamental vector fields of the elements of the basis $\{e_1,\ldots,e_n\}$ of the Lie algebra $(\mathbb{R}^n,+)$ be $\{\partial/\partial Q_1,\ldots, \partial/\partial Q_n\}$, respectively. 
	The space of leaves of the distribution spanned by the elements of $V^{HJ}$ is diffeomorphic to the manifold $M=\mathbb{R}^n$. Indeed, the variables $P_1,\ldots,P_n$ on $T^*\mathbb{R}^n\simeq \mathbb{R}^{2n}$ can be considered as a global coordinate system on $M$, which parametrises the leaves of the foliation $\mathcal{F}^{HJ}$. 
	
	The automorphic foliated Lie system related to (\ref{FLS1})  is, in virtue of Theorem \ref{FLSAFLS}, defined on the $(\mathbb{R}^n,+)$-principal bundle  $\pi: \mathbb{R}^n\times M\to  M$, $\pi:(\lambda,P) \mapsto P$, and it reads
	$$
	X(t,\lambda,P)=-\sum_{\alpha=1}^n\frac{\partial H}{\partial P_\alpha}(t,P)\frac{\partial}{\partial \lambda_\alpha}.
	$$
	
	Consider now the Lax pair given by (\ref{FLS2}) with 
	\begin{equation}\label{Casem}
		m(t,v)=\sum_{\alpha=1}^{n}\frac{\partial H}{\partial P_\alpha}(t,\lambda)e_\alpha.
	\end{equation}
	This particular value of $m(t,v)$ was studied in \cite{BV90} for $H$ being $t$-independent and it was shown to lead to a Lax pair for (\ref{FLS1}) under a simple change of variables. 
	
	The system (\ref{FLS3}) admits a Vessiot--Guldberg Lie algebra $V^{\rm ad}$ with a basis given by $$\left\{Y^s_1=2v^{n+1}\frac{\partial}{
	\partial v^1},\ldots,Y^s_n=2v^{2n}\frac{\partial}{\partial v^n}\right\}.
$$
 Such vector fields span an Abelian Lie algebra isomorphic to $(\mathbb{R}^n,+)$. The distribution spanned by the vector fields of $V^{\rm ad}$  on the submanifold of $\mathfrak{g}^{gl}$ of the form
$$
\mathcal{O}^{\rm ad}:=\left\{(v^1,\ldots,v^{2n})\in \mathfrak{g}^{gl}:\prod_{\alpha=n+1}^{2n}v^\alpha\neq 0\right\},
$$
gives rise to a family of leaves of the form (\ref{Leafg}) for $\prod_{i=1}^{n}k_{i}\neq 0$. Note indeed that $\mathcal{O}^{\rm ad}$ is the submanifold of $\mathfrak{g}^{gl}$ where the coordinates $\theta^i,I^i$ used in \cite{BV90} make sense.
	Therefore, the variables $v^{n+1},\ldots,v^{2n}$ can be considered as a coordinate system on the space of leaves, $M^{\rm ad}$, of $V^{\rm ad}$ within $\mathcal{O}^{\rm ad}$, which becomes a manifold locally diffeomorphic to $\mathbb{R}^n$. It is indeed an open subset of $\mathbb{R}^n$.
	
	The vector fields $Y^s_1,\ldots,Y^s_n$ can be integrated to obtain a local Lie group action
	$$
	\begin{array}{rccc}
		\varphi_{\mathfrak{g}^{gl}}:&\mathbb{R}^{n}\times \mathcal{O}^{\rm ad}&\longrightarrow& \mathcal{O}^{\rm ad},\\
		&(\lambda;v^1,\ldots,v^{2n})&\mapsto&(v^1-2\lambda_1v^{n+1},\ldots, v^n-2\lambda_nv^{2n},v^{n+1},\ldots,v^{2n}),
	\end{array}
	$$
	such that  the fundamental vector field of  $e_i$ in the canonical basis $\{e_1,\ldots,e_n\}$ of the Lie algebra $(\mathbb{R}^n,+)$ of $\mathbb{R}^n$  is $Y^s_i$ for $i=1,\ldots,n$.  
	The automorphic foliated Lie system associated with this foliated Lie system on $\mathcal{O}^{\rm ad}$, i.e. 
	\begin{equation}\label{FLS4}
	\frac{dv}{dt}=\sum_{\alpha=1}^ng_\alpha(t,v)X^{\mathfrak{g}^{lp}}_\alpha(v)=\sum_{\alpha=1}^nf_\alpha(t,v)Y^s_\alpha,\quad g_\alpha(t,v)=\frac{\partial H}{\partial P_{\alpha+n}}v^{\alpha+n},\quad \alpha=1,\ldots,n,
	\end{equation}
	reduces to the form of a $t$-dependent vector fields on $\mathbb{R}^{n}\times M^{\rm ad}$ of the form
	$$
	X(t,\lambda,v)=-\sum_{\alpha=1}^n\frac{\partial H}{\partial P_{\alpha+n}}(t,v^{n+1},\ldots,v^{2n})\frac{\partial}{\partial \lambda_\alpha}.
	$$
	Consequently, the solution to the Lax pair (\ref{FLS2}) on $\mathcal{O}^{\rm ad}$ for the particular value of $m(t,v)$ given in (\ref{Casem}) reduces to the same automorphic Lie system as the foliated Lie system (\ref{FLS1}) on the submanifold $\mathcal{H}$ of $T^*\mathbb{R}^n$ where $P_1\cdot\ldots\cdot P_n\neq 0$. Moreover, there exists a diffeomorphism $\phi:\mathcal{H}\subset T^*\mathbb{R}^n\rightarrow \mathcal{O}^{\rm ad}$ mapping (\ref{FLS1}) onto (\ref{FLS2}). It is easy to see that when two foliated Lie systems are diffeomorphic, they share the same foliated  automorphic Lie system. It is also immediate that foliated Lie systems related to the same automorphic foliated Lie system do not need to be diffeomorphic as they may be defined on manifolds of different dimension.
	\subsection{Stratified Lie--Hamilton systems and $r$-matrices}
	Let us illustrate the use of Poisson structures to investigate stratified Lie systems through a couple of examples. This suggests how to generalise the theory of Lie-Hamilton systems in \cite{CLS13} to the realm of stratified Lie systems. As a byproduct, several results on the use of $r$-matrices to study stratified Lie systems will be provided, which generalises previous findings from \cite{BV90}.
	
	Let $\{e_1,\ldots,e_{d}\}$ be a basis of the Lie algebra $\mathfrak{g}$ of a Lie group $G$ and let $\{v^1,\ldots,v^d\}$ be its dual basis. Consider the non-autonomous first-order system of differential equations on $\mathfrak{g}$ given by 
		 \begin{equation}\label{AdFLS}
		\frac{dx}{dt}=\sum_{\alpha=1}^{d}f_\alpha (t,x)X^{\rm ad}_{e_\alpha}(x),\qquad \forall x\in \mathfrak{g},\qquad \forall t\in \mathbb{R},
	\end{equation}
	where $X^{\rm ad}_{e_1},\ldots,X^{\rm ad}_{e_d}$ are the fundamental vector fields of the adjoint action of $G$ on $\mathfrak{g}$ induced by $\{e_1,\ldots,e_d\}$, respectively, and $f_1,\ldots,f_d$ are common $t$-dependent constants of motion for all the vector fields  $X^{\rm ad}_{e_1},\ldots,X^{\rm ad}_{e_d}$. Note that (\ref{AdFLS}) is a stratified Lie system on $\mathfrak{g}$.
	
	Let $x_0\in \mathcal{D}^{\rm ad}$  be a point where the rank of the distribution $\mathcal{D}^{\rm ad}$ spanned by $X^{\rm ad}_{e_1},\ldots,X_{e_d}^{\rm ad}$ reaches its maximum, $a_{\rm max}$, on $\mathfrak{g}$. Then, there exist $a_{\rm max}$ vector fields  taking values in $\mathcal{D}^{\rm ad}$ which are linearly independent at $x_0$.  Such vector fields will also be linearly independent at every point of a local neighbourhood $U$ of $x_0$, which causes  the rank of $\mathcal{D}^{\rm ad}$ to be $a_{\rm max}$ on $U$ (cf. \cite{Va94}). Hence,  (\ref{AdFLS}) may be restricted to an open submanifold $\mathcal{O}^{\rm ad}\subset \mathfrak{g}$, where $\mathcal{D}^{\rm ad}$ is a regular distribution of rank $a_{\max}$. Then, system (\ref{AdFLS}) becomes a foliated Lie system on $\mathcal{O}^{\rm ad}$. 
	
	Let us endow  $\mathfrak{g}$ with a Poisson structure so as to study (\ref{AdFLS}) via Poisson geometry techniques. This will provide an intrinsic geometric definition of the so-called {\it Kirillov bracket} on $\mathfrak{g}$ introduced in \cite{BV90} in an algebraic implicit manner. Assume that $\mathfrak{g}$ admits an  ${\rm ad}$-invariant non-degenerate constant metric ${\bf g}= {{\displaystyle\sum_{\alpha,\beta=1}^d}}g_{\alpha\beta}v^\alpha\otimes v^\beta$, i.e. ${\bf g}([x,x'],x'')+{\bf g}(x',[x,x'']) =0$ for all $x,x',x''\in \mathfrak{g}$ (see \cite{Mi76} for details on ad-invariant metrics). This allows us to define a metric tensor $\mathfrak{G}= {{\displaystyle\sum_{\mu,\nu=1}^{d}}}g_{\mu\nu}dv^\mu\otimes dv^\nu$ on $\mathfrak{g}$ and a vector bundle isomorphism $\mathfrak{G}^\flat:e_x\in T\mathfrak{g}\mapsto \mathfrak{G}_x(e_x,\cdot)\in T^*\mathfrak{g}$ between the tangent and cotangent bundles of $\mathfrak{g}$ for all $e_x\in T_x\mathfrak{g}$ and every $x\in \mathfrak{g}$. Since $\mathfrak{G}$ is non-degenerate, $\mathfrak{G}^\flat$ has an inverse $\mathfrak{G}^\sharp:T^*\mathfrak{g}\rightarrow T\mathfrak{g}$. Let $E$ be the Euler vector field on $\mathfrak{g}$ generating dilatations, namely $E= {{\displaystyle\sum_{\alpha=1}^d}}v^\alpha\partial/\partial v^\alpha$. It is immediate that $E$ does not depend on the dual basis in $\mathfrak{g}^*$ used to define it, which turns $E$ into a geometric object. 
	
 Define 
	\begin{equation}\label{KB}
		\{f,h\}_K=\mathfrak{G}([\mathfrak{G}^\sharp (df),\mathfrak{G}^\sharp (dh)]_\mathfrak{X},E),\qquad \forall f,h \in C^\infty(\mathfrak{g}),
	\end{equation}
	where $[\cdot,\cdot]_\mathfrak{X}$ is an extension of the Lie bracket on $\mathfrak{g}$ to $\mathfrak{X}(\mathfrak{g})$. More specifically, if $c_{\alpha\beta}\,^\gamma$, with $\alpha,\beta,\gamma=1,\ldots,d$, are the structure constants of the Lie algebra $\mathfrak{g}$ in the basis $\{e_1,\ldots,e_d\}$, i.e. $[e_\alpha,e_\beta]= {{\displaystyle\sum_{\gamma=1}^d}}c_{\alpha\beta}\,^\gamma e_\gamma$  for $\alpha,\beta=1,\ldots,d$, then $[f\partial/\partial v^\alpha,g\partial/\partial v^\beta]_\mathfrak{X}= {{\displaystyle\sum_{\gamma=1}^d}} fgc_{\alpha\beta}\,^\gamma \partial/\partial v^\gamma$ for $\alpha,\beta=1,\ldots,d$ and every $f,g\in C^\infty(\mathfrak{g})$. The $[\cdot,\cdot]_\mathfrak{X}$ on $\mathfrak{X}(\mathfrak{g})$ is sometimes called a {\it  bundle of Lie algebras}. Let us prove that (\ref{KB}) recovers the expression of the Kirillov bracket given in \cite{BV90}. 
Expression (\ref{KB}) is antisymmetric and satisfies the Leibniz property. Let us prove that (\ref{KB}) fulfils the Jacobi identity. Since ${\bf g}$ is non-degenerate and $\{e_1,\ldots,e_d\}$ is a basis of $\mathfrak{g}$, the linear functions  $f_\alpha: x\in\mathfrak{g}\mapsto {\bf g}(e_\alpha,x)\in \mathbb{R}$, with $\alpha=1,\ldots,d$, form a coordinate system on $\mathfrak{g}$. In the coordinate system $\{f_1,\ldots,f_d\}$, expression (\ref{KB}) becomes
	$
	\{f_\alpha,f_\beta\}_K(x)={\bf g}([e_\alpha,e_\beta],x)=f_{[e_\alpha,e_\beta]}(x)
	$  for $\alpha,\beta=1,\ldots,d$ and every $x\in \mathfrak{g}$. Then, 
	$
	\{f_\alpha,\{f_\beta,f_\gamma\}_K\}_K(x)={\bf g} ([e_\alpha,[e_\beta,e_\gamma]],x)
	$ for $\alpha,\beta,\gamma=1,\ldots,d$ and $x\in \mathfrak{g}$. It follows that (\ref{KB}) satisfies the Jacobi identity for  any three functions chosen among $f_1,\ldots,f_d$. Due to this and the fact that (\ref{KB}) satisfies the Leibniz property,  (\ref{KB}) obeys the Jacobi identity for all functions on $\mathfrak{g}$. Hence,  (\ref{KB}) becomes a Poisson bracket, and its Poisson bivector  reads
	\begin{equation}\label{Kirillov}
		\Lambda_K=\frac 12\sum_{\alpha,\beta=1}^d\Lambda_K(df_\alpha,df_\beta)\frac{\partial}{\partial f_\alpha}\wedge \frac{\partial}{\partial f_\beta}=\frac 12\sum_{\alpha,\beta,\gamma=1}^dc_{\alpha\beta}\,^\gamma f_\gamma\frac{\partial}{\partial f_\alpha}\wedge \frac{\partial}{\partial f_\beta}.
	\end{equation}

	Recall that the vectors $\{e_1,\ldots, e_d\}$ can be considered as a coordinate system on the dual space $\mathfrak{g}^*$. The Kirillov-Kostant-Souriau (KKS) bracket on $\mathfrak{g}^*$ reads  (see \cite{Va94} for details)
 $$
	\Lambda=\frac 12\sum_{\alpha,\beta, \gamma=1}^dc_{\alpha\beta}\,^\gamma e_\gamma \frac{\partial}{\partial e_\alpha}\wedge \frac{\partial }{\partial e_\beta}.
	$$
	The diffeomorphism $\phi:x\in \mathfrak{g}\mapsto {\bf g}(x,\cdot)\in \mathfrak{g}^*$ yields that $\Lambda=\phi_*\Lambda_K$. Hence, (\ref{KB}) is induced by the KKS bracket on $\mathfrak{g}^*$.
	
	
	Let us use the fact that ${\bf g}$ is ${\rm ad}$-invariant to prove that the vector fields $X^{\rm ad}_1,\ldots,X^{\rm ad}_d$ on $\mathfrak{g}$ are Hamiltonian relative to $\Lambda_K$. The ${\rm ad}$-invariance of ${\bf g}$ gives that ${\displaystyle \sum_{\delta=1}^d}c_{\alpha\beta}\,^\delta g_{\delta \gamma}=-{\displaystyle \sum_{\delta=1}^d}c_{\gamma\beta}\,^\delta g_{\delta\alpha}$ for $\alpha,\beta,\gamma=1,\ldots,d$. If $g^{\alpha\beta}$ are the entries of the inverse matrix of the metric ${\bf g}$ in the basis $\{e_1,\ldots,e_d\}$, one gets that 
	$$
	\sum_{\delta,\gamma,\alpha=1}^dg^{\theta\alpha}c_{\alpha\beta}\,^\delta g_{\delta\gamma}g^{\gamma\pi}=-\sum_{\gamma,\alpha,\delta=1}^dc_{\gamma\beta}\,^\delta g_{\delta\alpha}g^{\gamma\pi}g^{\theta\alpha}\Longrightarrow \sum_{\alpha=1}^dg^{\theta\alpha}c_{\alpha\beta}\,^\pi=-\sum_{\gamma=1}^dc_{\gamma\beta}\,^\theta g^{\gamma\pi}.
	$$
	Renaming indices and rewriting slightly,
${\displaystyle \sum_{\beta=1}^d}g^{\gamma \beta}c_{\alpha\beta}\,^\delta=-{\displaystyle \sum_{\beta=1}^d}g^{\delta\beta}c_{\alpha\beta}\,^\gamma$ for every $\gamma,\alpha,\delta=1,\ldots,d$. Since $v^\mu={\displaystyle \sum_{\gamma=1}^d}g^{\mu\gamma}f_\gamma$, using (\ref{Kirillov}), and in view of previous results, we have that 
	$$
	\Lambda_K=\frac 12\sum_{\alpha,\beta,\gamma,\mu,\nu,\sigma=1}^dc_{\alpha\beta}\,^\gamma g_{\gamma\nu} v^\nu g^{\mu\alpha}g^{\beta \sigma}\frac{\partial}{\partial v^\mu}\wedge\frac{\partial}{\partial v^\sigma}=-\frac 12\sum_{\beta,\sigma,\mu,\nu=1}^dg^{\beta\sigma}v^\nu c_{\nu\beta}\,^\mu\frac{\partial}{\partial v^\mu}\wedge \frac{\partial}{\partial v^\sigma}.
	$$
	A short calculation shows that $\Lambda_K(dv^{\bar{\sigma}},\cdot) ={\displaystyle\sum_{\beta=1}^d}g^{\beta\bar {\sigma}}X_{e_\beta}^{\rm ad}$ for $\bar\sigma=1,\ldots,d$, and 
	then 
	\begin{equation}\label{Eq:HamKil}
	X^{\rm ad}_{e_\beta}={ \Lambda_K\left(d\displaystyle\sum_{\bar\sigma=1}^dg_{\beta\bar{\sigma}}v^{\bar{\sigma}},\cdot \right)}=\Lambda_K(df_\beta,\cdot)
	\end{equation}
	for every $\beta=1,\ldots,d$. 
	This proves that the Vessiot--Guldberg Lie algebra of (\ref{AdFLS}) consists of Hamiltonian vector fields on {$\mathfrak{g}$}  relative to the  Kirillov bracket on $\mathfrak{g}$. The same applies to the restriction of the Vessiot--Guldberg Lie algebra of (\ref{AdFLS}) and $\Lambda_K$ to $\mathcal{O}^{\rm ad}$. It is worth noting that, in view of (\ref{Eq:HamKil}), the characteristic distribution of $\Lambda_K$ is spanned by $X^{\rm ad}_1,\ldots,X^{\rm ad}_d$. Hence, the symplectic leaves induced by the Poisson bracket $\Lambda_K$ are indeed the integral strata of the distribution spanned by $X^{\rm ad}_1,\ldots,X^{\rm ad}_d$ and $\Lambda_K$ can be restricted to such strata.
	
	Whether the Lax pair (\ref{AdFLS}) is  a Hamiltonian system or not relative to the Kirillov bracket on $\mathcal{O}^{\rm ad}$ is not much  relevant to us. In fact, it was proved in Section \ref{HOFSR} that our method to derive foliated superposition rules for (\ref{AdFLS}) requires to determine some common first-integrals for $X^{\rm ad}_{e_1},\ldots, X^{\rm ad}_{e_d}$ on the strata of $\mathcal{D}^{\rm ad}$ on $\mathcal{O}^{\rm ad}$ and their diagonal prolongations. 
	This can be achieved by using that these vector fields are Hamiltonian (see \cite{CLS13}). Moreover, (\ref{AdFLS}) can be restricted to the intersection of $\mathcal{O}^{\rm ad}$ with any leaf $\mathcal{S}$ of $\mathcal{D}^{\rm ad}$. Since $\mathcal{D}^{\rm ad}$ is regular with maximum rank on $\mathcal{O}^{\rm ad}$, a leaf $\mathcal{S}$ of $\mathcal{D}^{\rm ad}$, which always has a fixed dimension, is totally included in $\mathcal{O}^{\rm ad}$ or disjoint to it. We can also consider the restriction to some $\mathcal{S}\subset \mathcal{O}^{\rm ad}$ of  $X(t,x)={\displaystyle\sum_{\alpha=1}^d}f_\alpha(t,x)X^{\rm ad}_\alpha(x)$, where $x\in \mathcal{S}$, which is a Lie--Hamilton system.  Again, whether (\ref{AdFLS}) is a Hamiltonian system or not, {\it per se}, is not relevant.
	
	Let us consider now the automorphic foliated Lie system related to the foliated Lie system (\ref{SF}), considered on $\mathcal{O}^{\rm ad}$, relative to its Vessiot--Guldberg Lie algebra $V^{\rm ad}$. The analysis of this system will again support our idea about how one should endow stratified Lie systems with a compatible Poisson structure to study their properties. 
	
	Consider the Lie algebra $\mathfrak{g}^{lp}$ and its connected and simply connected Lie group $G^{LP}$. We consider the elements of $\mathfrak{g}^{lp}$ as left-invariant vector fields on $G^{LP}$ and we set $X^R_1,\ldots,X^R_{2n}$ to be the right-invariant vector field on $G^{LP}$ related to a basis $\{e_1,\ldots,e_{2n}\}$ of $\mathfrak{g}^{lp}$.  Let $r$ be an antisymmetric triangular $r$-matrix of $\mathfrak{g}^{lp}$, i.e. $r\in \Lambda^2\mathfrak{g}^{lp}$ and $[r,r]_{SN}=0$ with $[\cdot,\cdot]_{SN}$ being the Schouten-Nijenhuis bracket (see \cite{CP95,Ko97,Va94} for details), and define
	$$
	\Lambda_r=\frac 12\sum_{\alpha,\beta=1}^{2n}r^{\alpha\beta}X^R_\alpha\wedge X^R_\beta,
	$$
	where $r=\frac 12\sum_{\alpha,\beta=1}^{2n}r^{\alpha\beta}e_\alpha\wedge e_\beta$ and the vector fields $X^R_1,\ldots,X^R_{2n}$ satisfy that their non-zero commutation relations read $[X^R_{\alpha+n},X^R_{\alpha}]=-2X^R_\alpha$ for $\alpha=1,\ldots,n$.
	Then, $\Lambda_r$ is a Poisson bracket on $G^{LP}$ due to the fact that $r$ is a triangular $r$-matrix.  In particular,  consider the $r$-matrix $r={\displaystyle\sum_{\alpha=1}^n}e_\alpha\wedge h_{\alpha+n}$ in $\mathfrak{g}^{lp}$. This gives rise to a Poisson structure \begin{equation}\label{Eq:Pglp}\Lambda_r^{LP}={\displaystyle\sum_{\alpha=1}^n}X^R_\alpha\wedge X^R_{\alpha+n}
	\end{equation}
	 on $G^{LP}$.
	 
	Let us prove that the right-invariant vector fields $X_1^R,\ldots,X_n^R$ are  Hamiltonian relative to $\Lambda^{LP}_r$. Take the basis of right-invariant differential one-forms  $\{\eta_1^R,\ldots,\eta_{2n}^R\}$ on $G^{LP}$ that are dual  $\{X^R_1,\ldots,X^R_{2n}\}$. Then,
	$$
	d\eta^R_{\alpha}(X^R_\beta,X^R_\gamma)=-\eta^R_{\alpha}([X^R_\beta,X^R_\gamma]),\qquad \alpha,\beta,\gamma=1,\ldots,2n.
	$$
	Since $[X^R_\beta,X^R_\gamma]$ is a linear combination of $X^R_1,\ldots,X^R_n$, one gets that $d\eta^R_{\alpha}=0$ for $\alpha=n+1,\ldots,2n$. Since $G^{LP}$ is simply-connected, $\eta_\alpha^R=dk_\alpha$  for $\alpha=n+1,\ldots,2n$ and certain functions $k_{n+1},\ldots,k_{2n}$ on $\mathfrak{g}^{lp}$. Consequently, if follows from (\ref{Eq:Pglp}) that
	$$
	X^R_\alpha=-\Lambda_r^{LP}(dk_{\alpha+n},\cdot),\qquad \alpha=1,\ldots,n,
	$$
	are Hamiltonian vector fields relative to $\Lambda^{LP}_r$. 
	
	The $t$-dependent vector field  given by (\ref{SF}) on $\mathcal{O}^{\rm ad}$ is related via Theorem \ref{FLSAFLS}, when one considers that it admits the restriction of a Vessiot--Guldberg Lie algebra $V^{LP}\simeq \mathfrak{g}^{ lp}$ to $\mathcal{O}^{\rm ad}$, to the automorphic foliated Lie system 
	\begin{equation}\label{AutSpe}
		\frac{dg}{dt}=-\sum_{\alpha=1}^nf_\alpha(t,v_{n+1},\ldots,v_{2n})X^R_\alpha(g), \qquad g\in G^{LP},
	\end{equation}
	on the principal bundle $\pi^{PL}:G^{LP}\times M^{\rm ad}\rightarrow M^{\rm ad}$. This automorphic foliated Lie system
	admits a Vessiot--Guldberg Lie algebra $\langle X^R_1,\ldots, X^R_{n}\rangle $ of Hamiltonian vector fields relative to $\Lambda^{LP}_r$. Hence, a superposition rule relative to this Lie algebra can be obtained using the methods in \cite{CLS13}. Once again, one obtains that it is interesting to consider stratified Lie systems whose Vessiot--Guldberg Lie algebras are Hamiltonian relative to some Poisson bivector.
	
	Let us provide another example of foliated Lie system related to a Vessiot--Guldberg Lie algebra of Hamiltonian vector fields relative to a Poisson structure induced by a general $r$-matrix. 
	
	Recall that if $r=\frac 12{\displaystyle\sum_{\alpha,\beta=1}^{2n}}r^{\alpha\beta}e_\alpha\wedge e_\beta$ is  an $r$-matrix for $\mathfrak{g}^{lp}$, the {\it Sklyanin bracket} on $G^{LP}$ related to $r$ (see \cite{CP95}) is given by the Poisson bivector
	$$
	\Lambda^{S}=\frac 12\sum_{\alpha,\beta=1}^{2n}r^{\alpha\beta}(X^L_\alpha\wedge X^L_\beta-X^R_\alpha\wedge X^R_\beta).
	$$
	It can be proved that this Poisson bivector admits additionally properties (see \cite{CP95}), which justify to call it a {\it  Lie--Poisson bracket}. 
	It is well known that an $r$-matrix on $\mathfrak{g}$ induces a mapping $\delta_r:v\in \mathfrak{g}\mapsto [r,v]_{SN}\in \mathfrak{g}\wedge \mathfrak{g}$ whose transpose leads to a Lie algebra structure on $\mathfrak{g}^*$ and this, in turn, gives rise to a unique connected and simply connected Lie group $G^*$ related to $\mathfrak{g}^*$. Moreover, the $r$-matrix induces a unique Lie algebra structure, the {\it Drinfeld-double}, on the vector space $\mathfrak{d}=\mathfrak{g}\oplus \mathfrak{g}^*$ that is ${\rm ad}$-invariant,  reduces on $\mathfrak{g}$ to the original Lie algebra, and it reduces on $\mathfrak{g}^*$ to the Lie algebra induced by $r$ (see \cite{CP95} for details). Let $D$ be the connected and simply connected Lie group associated to the Lie algebra $\mathfrak{d}$. The embeddings $\mathfrak{g}\hookrightarrow \mathfrak{d}$ and $\mathfrak{g}^*\hookrightarrow \mathfrak{d}$ allow us to consider $G$ and $G^*$ as Lie subgroups of $D$. Moreover, every $f\in D$ in a close enough open neighbourhood $U$ of $e\in D$ can be written in a unique manner as the product of an element of $G$ by an element of $G^*$, in that order. In particular, $hg\in U$, with $h\in G^*$ and $g\in G$ admits a unique decomposition $hg=g^hh^g$ for $g^h\in G$ and $h^g\in G^*$.

	There exists for every element $\vartheta\in \mathfrak{g}^*$ a unique left- and right-invariant differential one-form on $G$, let us say $\eta^R_\vartheta,\eta^L_\vartheta$ respectively, whose values at $e$ are equal to $\vartheta$. Then, we define $X_\vartheta^{r,dr}=\Lambda^S(\eta^R_\vartheta,\cdot)$ and $X_\vartheta^{l,dr}=\Lambda^S(\eta^L_\vartheta,\cdot)$. It is known that the vector fields $X_\vartheta^{r,dr}$, with $\vartheta\in \mathfrak{g}^*$, span a Lie algebra $V^{dr}_l$ isomorphic to $\mathfrak{g}^*$. Its elements are called {\it left dressing vector fields}. The same applies to the vector fields $X_\vartheta^{l,dr}$, which generate a Lie algebra $V^{dr}_r$ isomorphic to $\mathfrak{g}^*$ whose elements are called {\it right dressing vector fields}. Both Lie algebras of vector fields can be integrated to obtain an, at least local, action of $G^*$ on $G$, the so-called {\it left} and {\it right dressing actions}, respectively. 
	
	Finally, let us define a last stratified Lie system related to a Poisson structure generalising autonomous Hamiltonian systems studied in \cite{CM10} (see that paper for details on further results).	Let us consider $\psi_{dr}:G^*\times G\rightarrow G$ with $\psi_{dr}(h,g)=g^h$ to be the left-dressing action\footnote{There exists a little misprint \cite[p. 1511, line 5]{CM10}, where it should be $\tilde h\in B$ instead of $\tilde h\in B^*$. This is a minor problem, but it can lead to a misunderstood in the following.}.  
	Consider the cotangent bundle $T^*G$, which is naturally diffeomorphic to $G\times \mathfrak{g}^*$ via the diffeomorphism $\vartheta_g \in T^*G\mapsto (g, T^*_eL_g\vartheta_g)\in G\times \mathfrak{g}^*$ for every $\vartheta_g\in T_g^*G$ and any $g\in G$. It is also well-known that $\psi_{dr}$ can be lifted to a new Lie group action $\widehat{\psi}_{dr}$ of $G^*$ on $T^*G\simeq G\times \mathfrak{g}^*$ (see \cite{AM87}). More particularly (see \cite[p. 1511, equation (1)]{CM10}), 
	$$
	\widehat{\psi}_{dr}(h,(g,\vartheta))=(\psi_{dr}(h,g),{\rm Ad}_{h^g}\vartheta),\quad \forall h\in G^*,\quad \forall g\in G,\quad \forall \vartheta\in \mathfrak{g}^*.
	$$
	This Lie group action has a momentum map $J:T^*G\simeq G\times \mathfrak{g}^*\rightarrow \mathfrak{g}$ obtained out of the fundamental vector fields of $\psi_{dr}$ as standardly known (see \cite{CM06} for details).
	Let $X^{dr}_1,\ldots,X^{dr}_r$ be a basis of fundamental vector fields for $\widehat{\psi}_{dr}$. One can define the $t$-dependent vector field on $G\times \mathfrak{g}^*$ of the form
	\begin{equation}\label{Eq:PseudoCol}
	X^{dr}(t,g,\vartheta)=\sum_{\alpha=1}^rf_\alpha(t,g,\vartheta)X^{dr}_\alpha(g,\vartheta),
	\end{equation}
	where $f_1(t,g,\vartheta),\ldots, f_r(t,g,\vartheta)$ are assumed to be common $t$-dependent constants of motion for $X^{dr}_1,\ldots,X^{dr}_r$, which are Hamiltonian vector fields with Hamiltonian functions $\langle J(g,\vartheta),e_\alpha\rangle$, with $\alpha=1,\ldots,r$, for a  basis $\{e_1,\ldots,e_r\}$ of $\mathfrak{g}$. Then, $X^{dr}$ is a stratified Lie system admitting a Vessiot--Guldberg Lie algebra of Hamiltonian vector fields. Systems of this type are related to collective Hamiltonian vector fields on $\mathfrak{g}^*$, which admit interesting applications (see \cite{CM10}). 
	
%
	

	Previous examples justify the following definition.

	\begin{definition} A {\it stratified Lie--Hamilton system} is a stratified Lie system $X$ on a manifold $N$ admitting a Vessiot--Guldberg Lie algebra of Hamiltonian vector fields relative to a Poisson structure on $N$.
	\end{definition}

	\section{Conclusions and Outlook}
	We have provided new applications of stratified and foliated Lie systems, which significantly extend the  examples given in \cite{CGM00}. We have introduced and studied foliated superposition rules and first-order systems of ODEs admitting foliated superposition rules. We have defined automorphic foliated Lie systems and their relations to foliated Lie systems have been analysed. 
	As applications, we have applied our techniques to a generalisation of Ermakov systems, we have illustrated our method to obtain foliated superposition rules, we have studied automorphic Lie systems related to Lax pairs and certain Hamiltonian systems, and the theory of Lie--Hamilton systems has been extended to stratified Lie--Hamilton systems.
	
	Our results can be extended to the so-called PDE Lie systems \cite{CGM07,Dissertationes,OG00,Ra11}. This can be accomplished by using the same fundamental ideas here depicted, but the development is technically much more complicated due to the nature of partial differential equations (PDEs). We are studying the possible applications of such a theory to physical models, which may justify their study. Moreover, we aim to look for generalisations of our ideas to collective systems related to (\ref{Eq:PseudoCol}), which could give rise to a generalisation of some results in \cite{CM10}. 
	
A natural generalisation of our techniques leads to analysing systems of ODEs given by a $t$-dependent vector field $X$ on a manifold $N$ so that $X$ is tangent a submanifold $\mathcal{S}\subset N$, where $X$ becomes a Lie system. This could lead to the analysis of  more general classes of systems of differential equations.

Several foliated Lie systems studied in the applications of this work are concerned with Hamiltonian systems admitting a maximal number of functionally independent autonomous constants of motion in involution relative to the Poisson bracket induced by the associated symplectic structure. Such systems can be understood as a $t$-dependent analogue of completely integrable Hamiltonian systems (see \cite{GMS02,Sa98,Sa12} for related notions). 
We aim to apply our methods as well as their possible generalisations to the analysis of $t$-dependent integrable non-commutative systems, which could be a generalisation of the previous ones admitting a maximal number of autonomous functionally independent constants of motion that need not be in involution (see \cite{Fa05,FT98,GMS02,LLV18,Ma12,SV00}). Finally, it is interesting to study how the theory of $r$-matrices can be applied to study the properties of certain Hamiltonian systems induced by them.

	\section*{Acknowledgements}
	D. Wysocki acknowledges partial financial support from the program Kartezjusz financed by the University of Warsaw and the Jagellonian University. Partial financial support by research projects  PGC2018-098265-B-C31 (MINECO, Madrid)  and DGA-E48/20R (DGA, Zaragoza)
	are acknowledged. J. de Lucas acknowledges funding
	from the Polish National Science Centre under grant HARMONIA 2016/22/M/ST1/00542.

\end{document}